\documentclass[9pt]{IEEEtran}
\usepackage{cite,physics,esdiff,mathtools}
\usepackage{amsmath,amssymb,amsfonts,amsthm,mathrsfs,comment,multirow,hyperref}
\usepackage[inline]{enumitem}
\usepackage[ruled,linesnumbered]{algorithm2e}
\SetKw{Break}{break}
\newtheorem{theorem}{Theorem}
\newtheorem{lemma}{Lemma}
\newtheorem{remark}{Remark}
\newtheorem{definition}{Definition}
\newtheorem{assumption}{Assumption}
\newtheorem{corollary}{Corollary}

\newtheorem{fact}{Fact}
\newtheorem{proc}{Procedure}
\usepackage{graphicx}

\usepackage{textcomp}
\def\BibTeX{{\rm B\kern-.05em{\sc i\kern-.025em b}\kern-.08em
    T\kern-.1667em\lower.7ex\hbox{E}\kern-.125emX}}

\DeclareMathOperator{\vect}{vec}
\DeclareMathOperator{\svec}{svec}
\begin{document}
\title{Reinforcement Learning for Adaptive Optimal Stationary Control of Linear Stochastic Systems}
\author{Bo Pang, \IEEEmembership{Student Member, IEEE} and Zhong-Ping Jiang, \IEEEmembership{Fellow, IEEE}
\thanks{This work has been supported in part by the U.S. National Science Foundation under Grants ECCS-1501044 and EPCN-1903781.}
\thanks{B. Pang and Z. P. Jiang are with the Control and Networks Lab, Department of Electrical and Computer Engineering, Tandon School of Engineering, New York University, 370 Jay Street, Brooklyn, NY 11201, USA (e-mail: bo.pang@nyu.edu; zjiang@nyu.edu).
}
}

\renewcommand{\arraystretch}{1.2}

\maketitle

\begin{abstract}
This paper studies the adaptive optimal stationary control of continuous-time linear stochastic systems with both additive and multiplicative noises, using reinforcement learning techniques. Based on policy iteration, a novel off-policy reinforcement learning algorithm, named optimistic least-squares-based policy iteration, is proposed which is able to find iteratively near-optimal policies of the adaptive optimal stationary control problem directly from input/state data without explicitly identifying any system matrices, starting from an initial admissible control policy. The solutions given by the proposed optimistic least-squares-based policy iteration are proved to converge to a small neighborhood of the optimal solution with probability one, under mild conditions. The application of the proposed algorithm to a triple inverted pendulum example validates its feasibility and effectiveness.
\end{abstract}

\section{Introduction}
Recently reinforcement learning (RL) and approximate/adaptive dynamic programming (ADP) have attracted the attention of both researchers from academia and practitioners from industry, due to their successful applications in the Chinese game of Go and video games \cite{sutton2018reinforcement}, multi-robot repair \cite{bhattacharya2020multiagent}, voltage regulation \cite{8763899}, human sensorimotor control \cite{Tao2020Bio,Jiang2014} and so on. Adaptive control \cite{astrom1995adaptive} is a field that deals with dynamical control systems with unknown parameters, but usually ignores the optimality of the control systems (with a few exceptions \cite{Bitmead1990Adaptive}). Optimal control \cite{Bryson1975optimal} is a branch of control theory that discusses the synthesis of feedback controllers to achieve optimality properties for dynamical control systems, but often requires the knowledge of the model parameters. Synthesizing the advantages of these two control methods \cite{126844}, RL/ADP searches for optimal controllers with respect to some performance index through interactions between the controller and the dynamical system, without the full knowledge of system dynamics \cite{sutton2018reinforcement,Bertsekas2019optimal}. Over the past decades, numerous RL and ADP methods have been proposed for different optimal control problems with various kinds of dynamical systems, see books \cite{sutton2018reinforcement,Bertsekas2019optimal,jiang2017robust,Dixon2018reinforcement} and recent surveys \cite{BUSONIU20188,8169685,powell2019reinforcement,levine2020offline,Jiang2020Survey} for details. However, most of existing RL and ADP methods are designed for stochastic discrete-time systems described by Markov decision processes \cite{sutton2018reinforcement,Bertsekas2019optimal} or deterministic continuous-time systems described by ordinary differential equations \cite{jiang2017robust,Dixon2018reinforcement}. Relatively fewer results are known for stochastic continuous-time systems described by stochastic differential equations \cite{pavliotis2014stochastic}, which are important in the modeling of stochastic uncertainties in physical systems, problems in finance and phenomenons in neuroscience, to new a few\cite{pavliotis2014stochastic}.

In this paper, a novel off-policy RL method named optimistic least-squares-based policy iteration (OLSbPI) is proposed, to find directly from input/state data near-optimal controllers solving the adaptive optimal stationary control problem, for continuous-time linear stochastic systems with additive and multiplicative noises. The optimal stationary control has been a classic problem in stochastic optimal control \cite{wonham1967dependent,wonham1968matrix,doi:10.1137/0311030,doi:10.1137/0309016,1099303}, whose objective is to find an optimal controller minimizing the expectation of a performance index with respect to the invariant probability measure of the closed-loop system. Dynamic programming, as a powerful tool for solving optimal control problems, has two main classes of algorithms: value iteration (VI) and policy iteration (PI) \cite{Bertsekas2019optimal}. Policy iteration involves two steps, \textit{policy evaluation} and \textit{policy improvement}. In the step of policy evaluation, a given controller is evaluated based on a scalar performance index. Then this performance index is utilized to generate a new controller in policy improvement. If all the system matrices are known and the two steps are iterated in turn, then the optimal solutions of the optimal stationary control problem are provably guaranteed to be found \cite{1099303}. In comparison, the proposed OLSbPI algorithm removes the restrictive requirement that the system matrices are known in policy iteration, which is achieved by directly implementing the policy evaluation and policy improvement steps from the input/state data, starting from an initial admissible control policy. Since the data is contaminated by unmeasurable stochastic noises, policy evaluation cannot be exactly implemented anymore because of the estimation errors occurring in each step of OLSbPI. Given that policy iteration is a nonlinear process, it is a non-trivial problem whether the policy iteration still converges to the optimal solution or its neighborhood in the presence of the estimation errors \cite{bertsekas2011approximate}. To this end, we firstly show that the policy iteration is robust to small estimation errors in the learning process (Theorem \ref{theorem_local_ISS} and Corollary \ref{corollary_global}), in the sense that as long as the estimation error is small, the solutions generated by policy iteration is still able to converge to a neighborhood of the optimal solution. Furthermore, under mild conditions, the solutions found by OLSbPI are proved to converge asymptotically to the optimal solutions (Theorem \ref{theorem_stochastic_algorithm_convergence}).

As compared with the past literature on similar topics, the proposed OLSbPI algorithm has several advantages. Firstly, the general case in which both additive and multiplicative stochastic noises are present in the systems are considered in the OLSbPI algorithm. Stochastic multiplicative noises are important in modeling the random perturbation in system parameters and coefficients, and are widely found in modern control systems such as networked control systems with noisy communication channels \cite{4118465}, modern power networks \cite{8814402}, neuronal brain networks \cite{Breakspear2017} and human sensorimotor control \cite{Jiang2014,Tao2020Bio}. Previous studies consider either only additive noises \cite{788532,doi:10.1137/18M1214147,basei2020linear} or deterministic cases \cite{bradtke1993reinforcement,8795815}. Secondly, starting from an initial admissible control policy, the proposed OLSbPI algorithm is able to find near-optimal solutions of the optimal stationary control problem directly from the input/state data, without the usage of noise-dependent information and the exact knowledge of all the system matrices. The optimal stationary control of linear stochastic systems with multiplicative noises are also investigated in \cite{wonham1967dependent,wonham1968matrix,doi:10.1137/0311030,doi:10.1137/0309016,1099303,6026952,7447723,Jiang2014,Tao2020Bio,Jiang2020Survey}. However to find the optimal policy, either the full knowledge of all the system matrices need to be known \cite{wonham1967dependent,wonham1968matrix,doi:10.1137/0311030,doi:10.1137/0309016,1099303}, or the knowledge of input matrix and gain matrices of the noises need to be known \cite{6026952}, or the stochastic noises are assumed measurable and explicitly used \cite{7447723,Jiang2014,Tao2020Bio,Jiang2020Survey}. Thirdly, the convergence of the OLSbPI algorithm is mathematically analyzed and rigorously proved. As a by-product of the convergence analysis, the policy iteration for optimal stationary control of linear stochastic systems is shown to be robust to the estimation errors in the learning process. It is well-known \cite{bertsekas2011approximate} that policy iteration may exhibit complex behavior when implemented inexactly, thus the derived robustness results of the policy iteration (Theorem \ref{theorem_local_ISS} and Corollary \ref{corollary_global}) are non-trivial and may be of independent interest. Finally, different from the RL/ADP methods designed for discrete-time control problems (with discrete state and action spaces), e.g. \cite{sutton2018reinforcement,bradtke1993reinforcement,8795815,lagoudakis2003least} and so on, our proposed OLSbPI algorithm deals with continuous-time control problems with continuous state and action spaces, and directly utilizes the continuous-time data flow.

This technical note is organized as follows. The optimal stationary control problem and related preliminaries are introduced in Section \ref{section_preliminaries}. The OLSbPI algorithm is derived in Section \ref{section_OLSbPI}, and its convergence analysis is given in Section \ref{section_Convergence_analysis}. A numerical example is employed to show the effectiveness of the proposed algorithm in Section \ref{section_simulation}. Finally, the conclusion is drawn in Section \ref{section_conclusion}.

\textbf{Notations.} $\mathbb{R}$ is the set of all real numbers; $\mathbb{Z}_+$ denotes the set of nonnegative integers; $\mathbb{S}^{n}$ is the set of all real symmetric matrices of order $n$; $\otimes$ denotes the Kronecker product; $I_n$ denotes the identity matrix of order $n$ while $0_{m\times n}$ denotes the $m\times n$ zero matrix; $\Vert \cdot\Vert_F$ is the Frobenius norm; $\Vert\cdot\Vert_2$ is the Euclidean norm for vectors and the spectral norm for matrices; for function $u:\mathbb{F}\rightarrow \mathbb{R}^{n\times m}$, $\Vert u\Vert_\infty$ denotes its $l^\infty$-norm when $\mathbb{F} = \mathbb{Z}_+$, and $L^\infty$-norm when $\mathbb{F} = \mathbb{R}$. For matrices $X\in \mathbb{R}^{m\times n}$, $Y\in \mathbb{S}^m$, and vector $v\in \mathbb{R}^{n}$, define 
$\vect(X) = [X_1^T, X_2^T, \cdots, X_n^T]^T$,
$\svec(Y) = [y_{11},\sqrt{2}y_{12},\cdots,\sqrt{2}y_{1m},y_{22},\sqrt{2}y_{23},\cdots,\sqrt{2}y_{m-1,m},y_{m,m}]^T\in \mathbb{R}^{\frac{1}{2}m(m+1)}$,
where $X_i$ is the $i$th column of $X$. $\vect^{-1}(\cdot)$ and $\svec^{-1}(\cdot)$ are operations such that $X=\vect^{-1}(\vect(X))$, and $Y=\svec^{-1}(\svec(Y))$. For $Z\in\mathbb{R}^{m\times n}$, define $\mathcal{B}_{r}(Z) = \{X\in\mathbb{R}^{m\times n}\vert \Vert X - Z\Vert_F<r\}$ and $\mathcal{\bar{B}}_{r}(Z)$ as the closure of $\mathcal{B}_{r}(Z)$. $Z^\dagger$ is the Moore-Penrose pseudoinverse of matrix $Z$. The direct sum of matrices $Z_1$ and $Z_2$ is denoted as $Z_1\oplus Z_2$.
\section{Problem Formulation and Preliminaries}\label{section_preliminaries}
Let $w_1\in\mathbb{R}^{q_1}$, $w_2\in\mathbb{R}^{q_2}$, $w_3\in\mathbb{R}^{q_3}$ be independent standard Brownian motions defined on a probability space $(\Omega,\mathcal{F},\mathcal{P})$. Consider the system described by the following stochastic differential equation:
\begin{equation}\label{LTI_sys_ito}
    \dd x = (Ax + Bu)\dd t + \sum_{j=1}^{q_1} D_jx\dd w_{1,j} + \sum_{k=1}^{q_2} F_ku\dd w_{2,k} + C\dd w_3,
\end{equation}
where $x\in\mathbb{R}^n$, $u\in\mathbb{R}^m$, $A$, $B$, $\{D_j\}_{j=1}^{q_1}$, $\{F_k\}_{k=1}^{q_2}$, $C$ are real constant matrices of compatible
dimensions. $(A,B)$ is controllable. Letting $u=\varphi(x)$, where $\varphi(\cdot)\colon \mathbb{R}^n \rightarrow \mathbb{R}^m$ is a Lipschitz continuous function on $\mathbb{R}^n$, and the random variable $x(0)$ be independent of $w_i$, $i=1,2,3$, notice that equation \eqref{LTI_sys_ito} is an equation of Ito's type, and determines a diffusion process $X_\varphi = \{x(t) \colon t\in\mathbb{R}, t\geq 0\}$. We are interested in the case where $X_\varphi$ has an invariant probability measure $\mu_\varphi$ defined on the Borel sets of $\mathbb{R}^n$; i.e., if $x(0)$ has the probability distribution $\mu_{\varphi}$, then so does $x(t)$, $\forall t>0$.
\begin{definition}[{\cite{doi:10.1137/0309016}}]
    Let $\Phi$ denote the class of admissible control policies, i.e., the set of functions $\varphi(\cdot)$ such that
    \begin{enumerate}
        \item $\varphi(\cdot)$ is Lipschitz continuous on its domain;
        \item an invariant probability measure $\mu_\varphi$ exists;
        \item for any invariant probability measure $\mu$,\\
        $\mathbb{E}_{\mu}[\Vert x\Vert_2^2] = \int_{\mathbb{R}^n} \Vert x\Vert_2^2\mu(dx)<\infty$.
    \end{enumerate}
\end{definition}
Given $A$, $B$, $\{D_j\}_{j=1}^{q_1}$, $\{F_k\}_{k=1}^{q_2}$ and $C$ in \eqref{LTI_sys_ito}, it is possible that set $\Phi$ is empty \cite{wonham1967dependent,wonham1968matrix,doi:10.1137/0309016,doi:10.1137/0311030}, since $\{\Vert D_j\Vert_F\}_{j=1}^{q_1}$ and $\{\Vert F_k\Vert_F\}_{k=1}^{q_2}$ may be so large that $X_{\varphi}$ always diverges to the infinity whatever the control policy $\varphi(\cdot)$ is, in which case no invariant measure exists. For $X\in\mathbb{R}^{n\times n}$, $Y\in\mathbb{R}^{m\times n}$, define
\begin{equation*}
\begin{split}
    \mathcal{L}_{Y}(X) &= (A-BY)^TX + X(A-BY) + \Pi(X) + Y^T\Sigma(X)Y, \\
    \mathcal{A}(Y) &= I_n\otimes (A-BY)^T + (A-BY)^T\otimes I_n \\
    &+ \sum_{j=1}^{q_1}D^T_j\otimes D^T_j + \sum_{k=1}^{q_2}(F_kY)^T\otimes(F_kY)^T,
\end{split}
\end{equation*}
where
$\Pi(X) = \sum_{j=1}^{q_1}D_j^TXD_j$, $\Sigma(X) = \sum_{k=1}^{q_2}F_k^TXF_k$. Then it is easy to check that
\begin{equation}\label{lyapunov_operator_equivalent}
    \vect(\mathcal{L}_{Y}(X)) = \mathcal{A}(Y)\vect(X).
\end{equation}
If $\mathcal{A}(Y)$ is Hurwitz, then $\mathcal{L}^{-1}_{Y}$ exists and for any $Z\in\mathbb{R}^{n\times n}$,
$$\mathcal{L}^{-1}_{Y}(Z) = \vect^{-1}\left(\mathcal{A}^{-1}(Y)\vect(Z)\right).$$
In fact, $\mathcal{L}^{-1}_{Y}(Z)$ is the unique solution of the generalized Lyapunov equation
$$\mathcal{L}_{Y}(X) = Z.$$
The following assumptions are made throughout this paper.
\begin{assumption}\label{Assumption_positive_definite}
$CC^T>0$ in system \eqref{LTI_sys_ito}.
\end{assumption}
\begin{assumption}\label{Assumption_admissible_nonempty}
There exists a $K\in\mathbb{R}^{m\times n}$, such that $\mathcal{A}(K)$ is Hurwitz.
\end{assumption}
\begin{fact}[{\cite[Section 2]{doi:10.1137/0309016}}]\label{fact_admissible_not_empty}
Under Assumption \ref{Assumption_admissible_nonempty}, $\Phi$ is not empty. If in addition Assumption \ref{Assumption_positive_definite} holds, then the invariant probability measure is unique for each admissible control policy.
\end{fact}
The following lemma gives a condition under which Assumption \ref{Assumption_admissible_nonempty} is satisfied.
\begin{lemma}[{\cite[Lemma]{1099303},\cite[Theorem 3.1]{doi:10.1137/0309016}}]\label{lemma_lyapunov}
For $K\in\mathbb{R}^{m\times n}$, $\mathcal{A}(K)$
is Hurwitz if and only if given any positive definite matrix $S$, the (unique) solution $P$ of
$\mathcal{L}_K(P) = -S $ is positive definite. A control policy $\varphi_K(x)=-Kx$ is admissible if $\mathcal{A}(K)$ is Hurwitz, in which case $K$ is also said to be admissible.
\end{lemma}
Since $(A,B)$ is controllable, there exists a $K$ such that $A-BK$ is Hurwitz. Then if $\{\Vert D_j\Vert_F\}_{j=1}^{q_1}$ and $\{\Vert F_k\Vert_F\}_{k=1}^{q_2}$ are sufficiently small, by Lemma \ref{math_fundamentals} in Appendix \ref{Appendix_auxiliary}, $\mathcal{A}(K)$ remains Hurwitz and Assumption \ref{Assumption_admissible_nonempty} holds, see \cite{wonham1967dependent} for details. Other conditions under which Assumption \ref{Assumption_admissible_nonempty} is satisfied can be found in \cite{doi:10.1137/0309016,doi:10.1137/0311030,WILLEMS1976277}.

In optimal stationary control \cite{wonham1967dependent,wonham1968matrix,doi:10.1137/0311030,doi:10.1137/0309016,1099303}, we want to find a $\varphi^*\in\Phi$, such that 
$$\mathbb{E}_{\mu_{\varphi^*}}[r(x,\varphi^*)] = \inf_{\varphi\in\Phi}\{\mathbb{E}_{\mu_\varphi}[r(x,\varphi)]\},$$
where $r(x,\varphi) = x^TQx + \varphi^TR\varphi$, $Q\in\mathbb{S}^n$ and $R\in\mathbb{S}^m$ are positive definite matrices. For $X\in\mathbb{R}^{n\times n}$, define
$$\mathcal{R}(X) = Q + A^TX + XA + \Pi(X) - XB(R+\Sigma(X))^{-1}B^TX.$$
\begin{lemma}[{\cite[Theorem 2]{1099303},\cite[Theorem 3.2]{doi:10.1137/0309016}}]\label{lemma_existence_of_optimal_control}
The optimal control policy is of the form $\varphi^*(x)=-K^*x$, and $\mathcal{A}(K^*)$ is Hurwitz, where
$$K^* = (R+\Sigma(P^*))^{-1}B^TP^*,$$
$P^*\in\mathbb{S}^n$ is the unique positive definite solution of
\begin{equation}\label{algebraic_riccati_equation}
    \mathcal{R}(P)=0,
\end{equation}
and the minimum cost is $\mathbb{E}_{\mu_{\varphi^*}}[r(x,\varphi^*)]=\trace(C^TP^*C)$.
\end{lemma}
For each linear state feedback control policy $\varphi_K(x) = -Kx$ with $\mathcal{A}(K)$ Hurwitz, by \cite[Theorem 3.1]{doi:10.1137/0309016}, the cost it induces is $$\mathbb{E}_{\mu_{\varphi_K}}[r(x,\varphi_K)]=\trace(C^TP_KC),$$
where $P_K\in\mathbb{S}^n$ is the unique positive definite solution of
\begin{equation}\label{algebraic_lyapunov_equation}
    \mathcal{L}_K(P_K) + Q + K^TRK = 0,
\end{equation}
and the subscript in $P_K$ is used to emphasize that $P_K$ is the solution of \eqref{algebraic_lyapunov_equation} associated with control gain $K$.

Define
\begin{equation*}
\begin{split}
G(P_K) &= \left[\begin{array}{c|c}
    [G(P_K)]_{xx} & [G(P_K)]^T_{ux} \\
    \hline
    [G(P_K)]_{ux} & [G(P_K)]_{uu}
\end{array}\right]\\
&\triangleq \left[\begin{array}{cc}
                    Q + A^TP_K + P_KA + \Pi(P_K) & P_KB \\
                    B^TP_K & R+\Sigma(P_K)
            \end{array}\right].
\end{split}
\end{equation*}
Then \eqref{algebraic_lyapunov_equation} is equivalent to
$$\mathcal{H}(G(P_K),K) = 0,$$
where
$\mathcal{H}(G(P_K),K) = \left[
    I_n , -K^T 
\right]G(P_K)\left[
    I_n, -K^T 
\right]^T$.
Since $P_K$ is positive definite if $\mathcal{A}(K)$ is Hurwitz, and $R$ is positive definite, $[G(P_K)]_{uu}$ is positive definite and invertible as long as $\mathcal{A}(K)$ is Hurwitz. The policy iteration to find the $P^*$ and $K^*$ in Lemma \ref{lemma_existence_of_optimal_control} is summarized in the following procedure.
\begin{proc}[Standard Policy Iteration]\label{procedure_policy_itration}
\ \par
\begin{enumerate}[label=\arabic*)]
    \item Choose a control gain $K_1$ with $\mathcal{A}(K_1)$ Hurwitz, and let $i=1$.
    \item\label{standard_PI_PE} (Policy evaluation) Evaluate the performance  of control gain $K_i$, by solving 
    \begin{equation}\label{RPI_PE}
       \mathcal{H}(G_i,K_i)=0
    \end{equation}
    for $P_i\in\mathbb{S}^n$, where $G_i=G(P_i)$.
    \item (Policy improvement) \label{standard_PI_PI}Get the improved policy by
    \begin{equation*}
        K_{i+1} = [G_{i}]_{uu}^{-1}[G_{i}]_{ux}.
    \end{equation*}
    \item Set $i\gets i+1$ and go back to Step \ref{standard_PI_PE}.
\end{enumerate}
\end{proc}
When there is only additive noise and one control-dependent noise term, i.e., $q_1 = 0$ and $q_2=1$ in system \eqref{LTI_sys_ito}, one can check that Procedure \ref{procedure_policy_itration} is actually a reformulation of the iterative method proposed in \cite[Theorem 1]{1099303}. By \cite[Remark 4), Remark 5)]{1099303}, the iterative method and its convergence properties in \cite[Theorem 1]{1099303} also apply to the general case involving multiple state-dependent and input-dependent noise terms as we considered here in system \eqref{LTI_sys_ito}. Thus we have the following convergence result for Procedure \ref{procedure_policy_itration}. 
\begin{theorem}[\cite{1099303}]\label{theorem_standard_PI}
    In Procedure \ref{procedure_policy_itration} we have:
    \begin{enumerate}[label=\roman*)]
        \item $\mathcal{A}(K_i)$ is Hurwitz for all $i=1,2,\cdots$.  \label{Hurwitz_matrix}
        \item $P_1\geq P_2 \geq P_3\geq \cdots \geq P^*$.\label{PI_monotone}
        \item $\lim_{i\rightarrow\infty}P_i=P^*$, $\lim_{i\rightarrow\infty}K_i = K^*$.\label{PI_converge}
    \end{enumerate}
\end{theorem}
Theorem \ref{theorem_standard_PI} guarantees that we can find the suboptimal approximations of $P^*$ and $K^*$ by implementing Procedure \ref{procedure_policy_itration} with sufficiently large number of iterations. However, the precise knowledge of system matrices is needed in Procedure \ref{procedure_policy_itration}, because $G_i$ depends on $A$, $B$, $\{D_j\}_{j=1}^{q_1}$, $\{F_k\}_{k=1}^{q_2}$, $Q$ and $R$. In this work, we assume that $A$, $B$, $\{D_j\}_{j=1}^{q_1}$, $\{F_k\}_{k=1}^{q_2}$, $C$, $Q$ and $R$ are all unknown, and propose a novel OLSbPI algorithm to directly estimate $G_i$ from input/state data, such that near-optimal solutions of the optimal stationary control problem can be found without explicitly identifying any system matrices in \eqref{LTI_sys_ito}. Due to the stochastic noises, the error arising for the estimation of each $G_i$ is unavoidable. Thus, to take into account the estimation errors, we propose the following procedure.
\begin{proc}[Robust Policy Iteration]\label{procedure_robust_policy_iteration}
\ \par
 \begin{enumerate}[label=\arabic*)]
        \item Choose a control gain $\hat{K}_1$ with $\mathcal{A}(\hat{K}_1)$ Hurwitz, and let $i=1$.
        \item\label{inexact_PI_PE} (Inexact policy evaluation) Obtain $\hat{G}_i = \Delta G_i + \tilde{G}_i$ (e.g., by approximately evaluating the performance of $\hat{K}_i$ directly from the input/state data, see Section \ref{section_OLSbPI}), where $\Delta G_i\in\mathbb{S}^{m+n}$ is a disturbance, $\tilde{G}_i=G(\tilde{P}_i)$, $\tilde{P}_i\in\mathbb{S}^n$ satisfies
        \begin{equation}\label{eRPI_PE}
            \mathcal{H}(\tilde{G}_i,\hat{K}_i)=0.
        \end{equation}
        \item (Policy update) Construct a new control gain
        \begin{equation*}
            \hat{K}_{i+1} = [\hat{G}_{i}]_{uu}^{-1}[\hat{G}_{i}]_{ux}.
        \end{equation*}
        \item Set $i\gets i+1$ and go back to Step \ref{inexact_PI_PE}.
    \end{enumerate}
\end{proc}

\section{Optimistic Least-squares-based Policy Iteration}\label{section_OLSbPI}
In this section, the optimistic least-squares-based policy iteration (OLSbPI) algorithm is proposed, which provides a concrete approach to construct estimation $\hat{G}_i$ in Procedure \ref{procedure_robust_policy_iteration} directly from input/state data. We firstly introduce the model-based optimistic policy iteration, then transform it into the model-free OLSbPI, Algrotihm \ref{algorithm_off_policy_stochastic}. The convergence analysis for Procedure 2 is given in Section \ref{section_Convergence_analysis} which yields the convergence of the OLSbPI algorithm. 

The optimistic policy iteration is based on the following result.
\begin{lemma}\label{Lemma_optimistic_PI}
For any control gain $K$ with $\mathcal{A}(K)$ Hurwitz, its associated $P_K$ satisfying
\eqref{algebraic_lyapunov_equation} is the unique stable equilibrium of linear dynamical system
\begin{equation}\label{ode_policy_evaluation_model}
    \dot{P} = \mathcal{H}(G(P),K), \quad P(0)\in\mathbb{S}^n
\end{equation}
\end{lemma}
\begin{proof}
By definition, \eqref{ode_policy_evaluation_model} can be rewritten as
$$\dot{P} = \mathcal{L}_K(P) + Q + K^TRK.$$
Vectorizing this equation and using \eqref{lyapunov_operator_equivalent}, we have
\begin{equation}\label{ode_policy_evaluation_model_vector}
    \vect(\dot{P}) = \mathcal{A}(K)\vect(P) + \vect(Q+K^TRK), \quad P(0)\in\mathbb{S}^{n}.
\end{equation}
Since $\mathcal{A}(K)$ is Hurwitz, obviously this linear dynamical system admits a unique stable equilibrium $P_K$.
\end{proof}
Lemma \ref{Lemma_optimistic_PI} implies that in policy evaluation, instead of solving \eqref{algebraic_lyapunov_equation}, one can solve the ODE \eqref{ode_policy_evaluation_model}. This is actually the continuous-time version of the optimistic policy iteration in \cite{tsitsiklis2002convergence,bertsekas2011approximate} for finite state and action spaces (thus the name ``optimistic''). Now, we show how \eqref{ode_policy_evaluation_model} in Lemma \ref{Lemma_optimistic_PI} is utilized to construct the estimation $\hat{G}_i$ in Procedure \ref{procedure_robust_policy_iteration} directly from input/state data, together with the least squares method. Suppose a control policy
\begin{equation}\label{input_collect_data_sys}
    u = -\hat{K}_1x + \sigma_uy
\end{equation}
is applied to the system \eqref{LTI_sys_ito}, where $\mathcal{A}(\hat{K}_1)$ is Hurwitz, $\sigma_u>0$ is a constant, $y\in\mathbb{R}^m$ is the exploration noise generated by stochastic differential equation
\begin{equation}\label{noise_sys}
    \dd y = -y\dd t + \dd w_4,
\end{equation}
and $w_4\in\mathbb{R}^m$ is standard Brownian motion independent of $w_i$, $i=1,2,3$. The cascaded system consisting of \eqref{LTI_sys_ito} and \eqref{noise_sys} is
\begin{equation}\label{Cascaded_LTI_sys}
\begin{split}
    \dd v &= \mathscr{A}v\dd t + \sum_{j=1}^{q_1}\mathscr{D}_jv\dd w_{1,j} + \sum_{k=1}^{q_2}\mathscr{F}_kv\dd w_{2,k} + \mathscr{C}\dd w_5.
    \end{split}
\end{equation} 
where $v = [x^T,y^T]^T$, $w_5 = [w^T_3,w^T_4]^T$, $\mathscr{D}_j = D_j\oplus 0_{m\times m}$, $\mathscr{C}=C\oplus I_m$ and
\begin{equation*}
\begin{split}
    \mathscr{A} &= \left[\begin{array}{cc}
         A-B\hat{K}_1 & \sigma_uB \\
         0_{m\times n} & -I_m
    \end{array}\right],\ \mathscr{F}_k =
    \left[\begin{array}{cc}
         -F_k\hat{K}_1 & \sigma_uF_k \\
         0_{m\times n} & 0_{m\times m}
    \end{array}\right].
\end{split}
\end{equation*}
Since $\mathcal{A}(\hat{K}_1)$ is Hurwitz, system \eqref{LTI_sys_ito} with control policy $\varphi(x)=-\hat{K}_1x$ is mean square stable \cite[Definition 1.]{WILLEMS1976277}, in the absence of the noise term $w_3$ \cite{1099206}. Obviously, system \eqref{noise_sys} is also mean square stable, in the absence of the noise term $w_4$. Then by the stochastic small gain theorem \cite{DRAGAN1997243,1712310}, the cascaded system \eqref{Cascaded_LTI_sys} is mean square stable, in the absence of $w_5$. Thus control policy \eqref{input_collect_data_sys} is admissible \cite[Equation (6)]{1099206}, and induces a unique invariant probability measure $\mu_{v}$ on the Borel sets of $\mathbb{R}^n\times \mathbb{R}^m$ for the cascaded system \eqref{Cascaded_LTI_sys}.
\begin{assumption}\label{assumption_moments}
For some integer $p\geq 4$, $\mathbb{E}_{\mu_v}[\Vert v\Vert^p_2]<\infty$.
\end{assumption}
\begin{remark}\label{remark_moment_exist}
According to \cite[Lemma 4.1]{haussmann1974existence}, Assumption \ref{assumption_moments} holds if for \eqref{Cascaded_LTI_sys}
$$\left\Vert\int_0^{\infty} e^{t\mathscr{A}^T}\Gamma(D_1,\cdots,D_{q_1}, F_1,\cdots,F_{q_2})e^{t\mathscr{A}}\dd t\right\Vert_2<\frac{1}{p-1},$$
where $\Gamma(\cdot)$ is a continuous matrix-valued function of $\{D_j\}_{j=1}^{q_1}$ and $\{F_k\}_{k=1}^{q_2}$, and vanishes at the origin. Since $\mathscr{A}$ is Hurwitz, the above inequality is satisfied when $\{\Vert D_j\Vert_F\}_{j=1}^{q_1}$ and $\{\Vert F_k\Vert_F\}_{k=1}^{q_2}$ are small. When there is only additive noise, i.e., $D_j=0$ and $F_k=0$ for all $j$ and $k$, the above equation holds automatically and thus Assumption \ref{assumption_moments} holds for arbitrarily large $p$. See \cite{haussmann1974existence} for details.
\end{remark}
For any $P\in\mathbb{S}^{n}$, Ito's formula \cite[Lemma 3.2]{pavliotis2014stochastic} yields
\begin{equation}\label{ADP_noise_dependent_data_eqn_basic}
\begin{split}
    \dd (x^TPx) &= 2x^TP(Ax+Bu)\dd t + x^T\Pi(P)x\dd t \\
    &+ u^T\Sigma(P)u\dd t + \trace(C^TPC)\dd t + 2x^TP\dd w,
\end{split}
\end{equation}
where $\dd w = \sum_{j=1}^{q_1} D_jx\dd w_{1,j} + \sum_{k=1}^{q_2} F_ku\dd w_{2,k} + C\dd w_3$.
By vectorization, from \eqref{ADP_noise_dependent_data_eqn_basic} we have
\begin{equation}\label{ADP_dependent_noise_data_eqn_differential}
\begin{split}
    \tilde{z}\dd(\tilde{x}^T)\svec(P) &= \tilde{z}\tilde{z}^T\svec(\theta(P))\dd t \\
    &- \tilde{z}r(x,u)\dd t + 2\tilde{z}x^TP\dd w,
\end{split}
\end{equation}
where $z=[x^T,u^T,1]^T$, $\tilde{z} = \svec(zz^T)$, $\tilde{x} = \svec(xx^T)$ and $\theta(P) = G(P)\oplus\trace(C^TPC)$. Integrating \eqref{ADP_dependent_noise_data_eqn_differential} from 0 to $t_f>0$ yields,
\begin{equation}\label{data_eqn_proto}
\begin{split}
&\psi_{t_f}\svec(\theta(P)) = \zeta_{t_f}\svec(P) + \xi_{t_f} - \eta_{t_f},
\end{split}
\end{equation}
where
\begin{equation*}
\begin{split}
\psi_{t_f} &= \frac{1}{t_f}\int_0^{t_f}\tilde{z}\tilde{z}^T\dd t,\quad \zeta_{t_f} = \frac{1}{t_f}\int_0^{t_f}  \tilde{z}\dd(\tilde{x}^T),\\
\xi_{t_f} &= \frac{1}{t_f}\int_0^{t_f}\tilde{z}r(x,u)\dd t,\quad
\eta_{t_f} = \frac{1}{t_f}\int_0^{t_f}2\tilde{z}x^TP\dd w.
\end{split}
\end{equation*}
By Assumption \ref{assumption_moments}, Birkhoff ergodic theorem \cite[Theorem 16.14]{korlov2007theory}, and \cite[the lemma on page 530]{lee_kozin_1977}, the following relationships hold almost surely:
\begin{equation}\label{ergodic_converge_1}
\begin{split}
    \lim_{t_f\rightarrow\infty} \psi_{t_f} &= \Psi \triangleq \mathbb{E}_{\mu_{v}}[\tilde{z}\tilde{z}^T],\quad \lim_{t_f\rightarrow\infty} \eta_{t_f} = 0,\\ \lim_{t_f\rightarrow\infty} \xi_{t_f} &= \Xi \triangleq \mathbb{E}_{\mu_{v}}[\tilde{z}r(x,u)].
\end{split}
\end{equation}
From the fact that \eqref{data_eqn_proto} holds for any $P\in\mathbb{S}^{n}$, almost surely
\begin{equation}\label{ergodic_converge_2}
    \lim_{t_f\rightarrow\infty} \zeta_{t_f} = \mathcal{Z},
\end{equation}
where $\mathcal{Z}$ is a constant matrix.
\begin{assumption}\label{PE_assumption}
$\Psi$ is nonsingular.
\end{assumption}
Combining Assumption \ref{PE_assumption}, \eqref{data_eqn_proto}, \eqref{ergodic_converge_1} and \eqref{ergodic_converge_2} gives
\begin{equation}\label{exact_theta}
    \svec(\theta(P)) = \Psi^{-1}\left(\mathcal{Z}\svec(P) + \Xi\right).
\end{equation}
Since $G(P) = \mathcal{H}(\theta(P),0)$, \eqref{ode_policy_evaluation_model} is identical to the following dynamical system
\begin{equation}\label{ode_policy_evaluation_data_precise}
    \dot{P} = \mathcal{H}(\mathcal{H}(\svec^{-1}(\Psi^{-1}\left(\mathcal{Z}\svec(P) + \Xi\right)),0),K),
\end{equation}
where $P(0)\in\mathbb{S}^n$. The OLSbPI is presented in Algorithm \ref{algorithm_off_policy_stochastic}. By \eqref{ergodic_converge_1} and \eqref{ergodic_converge_2}, Line \ref{algorithm_policy_evaluation} of Algorithm \ref{algorithm_off_policy_stochastic} is an approximation of \eqref{ode_policy_evaluation_data_precise}, while Line \ref{algorithm_estimate_final_Q_1} of Algorithm \ref{algorithm_off_policy_stochastic} is an approximation of \eqref{exact_theta} with $P$ replaced by $\hat{P}(s_f)$. Thus the solution of the ODE at $s_f$ in Line \ref{algorithm_policy_evaluation} of Algorithm \ref{algorithm_off_policy_stochastic} is an estimation of $\tilde{P}_i$ satisfying \eqref{eRPI_PE}, while $\hat{\theta}_i$ given by Line \ref{algorithm_estimate_final_Q_1} of Algorithm \ref{algorithm_off_policy_stochastic} is an estimation of $\theta(\tilde{P}_i)$. By definition of $\theta(\cdot)$ in \eqref{ADP_dependent_noise_data_eqn_differential}, the estimation $\hat{G}_i$ of $\tilde{G}_i$ satisfying \eqref{eRPI_PE} is given by Line \ref{algorithm_estimate_final_Q_2} of Algorithm \ref{algorithm_off_policy_stochastic}. Notice that the same data matrices $\psi_{t_f}$, $\zeta_{t_f}$ and $\xi_{t_f}$ are reused for all iterations, thus OLSbPI is off-policy. In addition, only state $x$ and input $u$ appear in data matrices $\psi_{t_f}$, $\zeta_{t_f}$ and $\xi_{t_f}$. Therefore OLSbPI does not explicitly use the noise information, which is different from the methods proposed in \cite{7447723,Jiang2014,Tao2020Bio,Jiang2020Survey}.
\begin{remark}
Assumption \ref{PE_assumption} is in the spirit of persistent excitation condition in adaptive control \cite{astrom1995adaptive}. Similar assumptions are needed in other RL methods, see \cite{doi:10.1137/18M1214147,jiang2017robust,Dixon2018reinforcement,8169685,7447723,Jiang2014,Tao2020Bio}. Assumption \ref{PE_assumption} makes the data-based differential equation \eqref{ode_policy_evaluation_data_precise} equivalent to the model-based differential equation \eqref{ode_policy_evaluation_model}, which is a key component in the convergence analysis of the next section. 
\end{remark}
\begin{remark}\label{remark_explore}
The presence of exploration noise $y$ in \eqref{input_collect_data_sys} is necessary for Assumption \ref{PE_assumption} to hold. Otherwise $u$ will always be linearly dependent on $x$ and, as a result, $\Psi$ cannot be nonsingular. To see this, consider the case where both $x$ and $u$ are scalars, and $\sigma_u=0$ in \eqref{input_collect_data_sys}. By definition, $\tilde{z} = \left[x^2, \sqrt{2}\hat{K}_1x^2, \sqrt{2}x, \hat{K}_1^2x^2, \sqrt{2}\hat{K}_1x, 1 \right]^T$. Then it is easy to check that the first two rows of $\Psi = \mathbb{E}_{\mu_v}[\tilde{z}\tilde{z}^T]$ are linearly dependent. This is also true for the case where both $x$ and $u$ are vectors.
\end{remark}
\begin{algorithm}
\KwIn{Initial control gain $\hat{K}_1$ with $\mathcal{A}(\hat{K}_1)$ Hurwitz, Number of policy iterations $N$, Length of policy evaluation $s_f$, Length of rollout $t_f$, Exploration noise magnitude $\sigma_u$.}
Apply control policy \eqref{input_collect_data_sys} to system \eqref{LTI_sys_ito} to generate a trajectory of input/state data of length $t_f$\;
Construct data matrices $\psi_{t_f}$, $\zeta_{t_f}$ and $\xi_{t_f}$ defined in \eqref{data_eqn_proto}\;
\For{$i=1,\cdots,N-1$}{
$\hat{P}_i(0)\leftarrow 0$\label{algorithm_policy_evaluation_start}\;
Solving the following ODE on $[0,s_f]$:\\
$\dot{\hat{P}}_i = \mathcal{H}(\mathcal{H}(\svec^{-1}(\psi_{t_f}^\dagger(\zeta_{t_f}\svec(\hat{P}_i) + \xi_{t_f}),0),\hat{K_i})$\label{algorithm_policy_evaluation}\;
$\hat{\theta}_{i}\leftarrow\svec^{-1}(\psi_{t_f}^\dagger(\zeta_{t_f}\svec(\hat{P}_i(s_f))+ \xi_{t_f})$\;\label{algorithm_estimate_final_Q_1}
$\hat{G}_{i}\leftarrow \mathcal{H}(\hat{\theta}_{i},0)$\;\label{algorithm_estimate_final_Q_2}
$\hat{K}_{i+1}\leftarrow [\hat{G}_{i}]_{uu}^{-1}[\hat{G}_{i}]_{ux}$\;
}
\KwRet{$\hat{K}_N$.}
\caption{OLSbPI\label{algorithm_off_policy_stochastic}}
\end{algorithm}
\section{Convergence Analysis}\label{section_Convergence_analysis}
In this section, we first show in Corollary \ref{corollary_global} that whenever $\Vert\Delta G_i\Vert_\infty$ is small, $\tilde{P}_i$ in Procedure \ref{procedure_robust_policy_iteration} (the solution of \eqref{eRPI_PE}) is bounded, and, moreover, enters and stays in a small neighborhood of $P^*$. Then we show that $\Vert\Delta G_i\Vert_\infty$ in Algorithm \ref{algorithm_off_policy_stochastic} can be made small by choosing $s_f$ and $t_f$ large enough (Lemma \ref{Lemma_policy_evaluation_error_bounded}), which completes the convergence analysis of OLSbPI (Theorem \ref{theorem_stochastic_algorithm_convergence}).
\subsection{Convergence Analysis of Procedure \ref{procedure_robust_policy_iteration}}
For $P\in\mathbb{S}^n$ define
$$\mathscr{K}(P) = \mathscr{R}(P)^{-1}B^TP,\quad \mathscr{R}(P) = R + \Sigma(P).$$
Suppose in Procedure \ref{procedure_policy_itration}, $K_1 = \mathscr{K}(P_0)$ and $\mathcal{A}(K_1)$ is Hurwitz, where $P_0\in\mathbb{S}^n$. Such a $K_1$ always exists, for example, when $P_0$ is close to $P^*$ by Lemma \ref{math_fundamentals} in Appendix \ref{Appendix_auxiliary}. From Theorem \ref{theorem_standard_PI} and \eqref{lyapunov_operator_equivalent}, sequence $\{P_i\}_{i=0}^{\infty}$ generated by Procedure \ref{procedure_policy_itration} satisfies
\begin{equation}\label{Kleinman_nonlinear_system}
    \vect(P_{i+1}) = \mathcal{A}^{-1}(\mathscr{K}(P_i))\vect(-Q-\mathscr{K}(P_i)^TR\mathscr{K}(P_i)).
\end{equation}
\begin{lemma}\label{lemma_invertible_hurwitz}
There exists a $\delta_0>0$, such that $\mathscr{R}(P)$ is positive definite, and $\mathcal{A}(\mathscr{K}(P))$ is Hurwitz for all $P\in\bar{\mathcal{B}}_{\delta_0}(P^*)$.
\end{lemma}
\begin{proof}
Since $R>0$, $P^*>0$ and $\mathcal{A}(\mathscr{K}(P^*))$ is Hurwitz by Lemma \ref{lemma_existence_of_optimal_control}, the lemma is proved using Lemma \ref{math_fundamentals} in Appendix \ref{Appendix_auxiliary}.
\end{proof}
The next lemma shows that $P^*$ is a locally exponentially stable equilibrium \cite[Definition 2.5]{JIANG201858} of nonlinear system \eqref{Kleinman_nonlinear_system}.
\begin{lemma}\label{lemma_local_exponentail_stable}
    For any $\sigma<1$, there exists a $\delta_1(\sigma)\in(0,\delta_0]$ and a $C_1(\delta_1)>0$, where $\delta_0$ is defined in Lemma \ref{lemma_invertible_hurwitz}, such that for any $P_i\in\mathcal{B}_{\delta_1}(P^*)$,
    \begin{enumerate}[label=\roman*)]
    \item\label{feasibility} $\mathcal{A}(\mathscr{K}(P_i))$ is Hurwitz, $[G_{i}]_{uu}$ is invertible.
    \item\label{quadratic_rate} Procedure \ref{procedure_policy_itration} has a local quadratic convergence rate, i.e.,
    $$\Vert P_{i+1}-P^* \Vert_F \leq C_1\Vert P_i-P^*\Vert_F^2.$$
    \item\label{local_exp} \eqref{Kleinman_nonlinear_system} is locally exponentially stable at $P^*$
    $$\Vert P_{i+1}-P^* \Vert_F \leq \sigma\Vert P_i-P^*\Vert_F.$$
    \end{enumerate}
\end{lemma}
\begin{proof}
Item \ref{feasibility} follows directly from Lemma \ref{lemma_invertible_hurwitz}. Subtracting $$K_{i+1}^TB^TP^*+P^*BK_{i+1}-K_{i+1}^T(R+\Sigma(P^*))K_{i+1}$$ from both sides of \eqref{algebraic_riccati_equation} yields
\begin{equation}\label{Kleinman_eqn_1}
\begin{split}
    &\mathcal{A}(K_{i+1})\vect(P^*) = \vect(-Q-K_{i+1}^TRK_{i+1}) \\
    &+ \vect((K_{i+1}-K^*)^T\mathscr{R}(P^*)(K_{i+1}-K^*)), 
\end{split}
\end{equation}
Since $K_{i+1}=\mathscr{K}(P_i)$, subtracting \eqref{Kleinman_eqn_1} from \eqref{Kleinman_nonlinear_system} yields
\begin{align}
    \vect(P_{i+1}-P^*) &= -\mathcal{A}^{-1}(\mathscr{K}(P_i))\vect((\mathscr{K}(P_i)-\mathscr{K}(P^*))^T \nonumber\\
    &\times\mathscr{R}(P^*)(\mathscr{K}(P_i)-\mathscr{K}(P^*))).    \label{Kleinman_difference}
\end{align}
By \eqref{linear_perturbation} in Appendix \ref{Appendix_auxiliary}, we have for all $P_i\in\bar{\mathcal{B}}_{\delta_0}(P^*)$,
\begin{equation*}
    \begin{split}
        &\Vert\mathscr{K}(P_i)-\mathscr{K}(P^*)\Vert_F
        \leq \Vert \mathscr{R}^{-1}(P_i)\Vert_F\Vert B^T\Vert_F\Vert P_i-P^*\Vert_F \\
        &+ \Vert \mathscr{R}^{-1}(P_i)\Vert_F \Vert \Sigma(P_i-P^*)\Vert_F \Vert\mathscr{R}^{-1}(P^*)\Vert_F\Vert B^TP^*\Vert_F \\
        &\leq L_{\mathscr{K}}\Vert P_i-P^*\Vert_F,
    \end{split}
\end{equation*}
where 
$$L_{\mathscr{K}} \triangleq L_I\left(\Vert B^T\Vert_F + \Vert\mathscr{R}^{-1}(P^*)\Vert_F\Vert B^TP^*\Vert_F + \sum_{j=1}^{q_1}\Vert D_j\Vert^2_F \right) $$ and $L_I>0$ is an upper bound of continuous function $\Vert \mathscr{R}^{-1}(\cdot)\Vert_F$ on compact set $\bar{\mathcal{B}}_{\delta_0}(P^*)$ (see Lemma \ref{math_fundamentals} in Appendix \ref{Appendix_auxiliary}). Thus, \eqref{Kleinman_difference} implies for all $P_i\in\bar{\mathcal{B}}_{\delta_0}(P^*)$,
\begin{align}
    \Vert P_{i+1}-P^*\Vert_F 
    &\leq \Vert\mathcal{A}^{-1}(\mathscr{K}(P_i))\Vert_F\Vert\mathscr{R}(P^*)\Vert_F L^2_{\mathscr{K}}\Vert P_i-P^*\Vert_F^2 \nonumber \\
    &\leq C_1\Vert P_i-P^*\Vert_F^2, \nonumber
\end{align}
where again the fact that continuous function $\Vert\mathcal{A}^{-1}(\mathscr{K}(P_i))\Vert_F$ is upper bounded on compact set $\bar{\mathcal{B}}_{\delta_0}(P^*)$ is used. This proves Item \ref{quadratic_rate}. Choosing $\delta_1\in(0,\delta_0]$ small enough such that
$C_1\Vert P_i-P^*\Vert_F\leq \sigma$ proves Item \ref{local_exp}. 
\end{proof}
Suppose in Procedure \ref{procedure_robust_policy_iteration},
$\hat{K}_1 = \mathscr{K}(\tilde{P}_0)$ and $\Delta G_0 = 0$,
where $\tilde{P}_0\in\mathbb{S}^n$ is chosen such that $\mathcal{A}(\hat{K}_1)$ is Hurwitz. If $[\hat{G}_{i}]_{uu}$ is invertible and $\mathcal{A}(\hat{K}_i)$ is Hurwitz for all $i$ (which will be proved later), the sequence $\{\tilde{P}_i\}_{i=0}^{\infty}$ generated by Procedure \ref{procedure_robust_policy_iteration} satisfies
\begin{equation}\label{Robust_Policy_Iteration_nonlinear_system}
\begin{split}
    \vect(\tilde{P}_{i+1})  &= \mathcal{A}^{-1}(\mathscr{K}(\tilde{P}_i))\vect(-Q-\mathscr{K}(\tilde{P}_i)^TR\mathscr{K}(\tilde{P}_i)) \\
    &+ \mathcal{E}(\tilde{G}_i,\Delta G_i),
\end{split}
\end{equation}
where
\begin{equation*}
\begin{split}
    \mathcal{E}(\tilde{G}_i,\Delta G_i) &= \mathcal{A}^{-1}(\mathscr{K}(\tilde{P}_i))\vect(Q+\mathscr{K}(\tilde{P}_i)^TR\mathscr{K}(\tilde{P}_i)) \\
    &- \mathcal{A}^{-1}(\hat{K}_{i+1})\vect(Q+\hat{K}_{i+1}^TR\hat{K}_{i+1}).
\end{split}
\end{equation*}
Based on Lemma \ref{lemma_local_exponentail_stable}, the following theorem is obtained, which states that nonlinear system \eqref{Robust_Policy_Iteration_nonlinear_system} is locally input-to-state stable \cite[Definition 2.1]{JIANG201858} at $P^*$, if $\Delta G_i$ is regarded as the input. The proof is given in Appendix \ref{section_proof_local_ISS}. 
\begin{theorem}\label{theorem_local_ISS}
For $\sigma$ and its associated $\delta_1$ in Lemma \ref{lemma_local_exponentail_stable}, there exists $\delta_3(\delta_1)>0$, such that if $\Vert\Delta G\Vert_\infty<\delta_3$, $\tilde{P}_0\in\mathcal{B}_{\delta_1}(P^*)$,
\begin{enumerate}[label=(\roman*)]
    \item\label{theorem_local_ISS_item_well_defined} $[\hat{G}_{i}]_{uu}$ is invertible and $\mathcal{A}(\hat{K}_i)$ is Hurwitz, $\forall i\in\mathbb{Z}_+$, $i>0$;
    \item\label{theorem_local_ISS_item_local_ISS} the following local input-to-state stability holds:
    $$\Vert \tilde{P}_i - P^*\Vert_F\leq \beta(\Vert\tilde{P}_0-P^*\Vert_F,i) + \gamma(\Vert \Delta G\Vert_\infty),$$
    where for any $q\in\mathbb{R}$, $\beta(q,i) = \sigma^iq$, $\gamma(q) = qC_2/(1-\sigma)$ and $C_2(\delta_3)>0$.
    \item\label{theorem_local_ISS_item_K_bounded} $\Vert\hat{K}_i\Vert_F<\kappa^b_1$ for some $\kappa^b_1>0$, $\forall i\in\mathbb{Z}_+$, $i>0$;
    \item\label{theorem_local_ISS__item_converging_input_converging_state} $\lim_{i\rightarrow\infty} \Vert\Delta G_i\Vert_F = 0$ implies $\lim_{i\rightarrow\infty} \Vert \tilde{P}_i-P^* \Vert_F=0$.
\end{enumerate}
\end{theorem}
Intuitively, Theorem \ref{theorem_local_ISS} implies that in Procedure \ref{procedure_robust_policy_iteration} if $\tilde{P}_0$ is close to $P^*$ (thus $\hat{K}_1$ is close to $K^*$), and the disturbance input $\Delta G$ is bounded and not too large, then the cost of the generated control policy $\hat{K}_i$ is also bounded, and will ultimately be no larger than a constant proportional to the $l^\infty$-norm of the disturbance. The smaller the disturbance is, the better the ultimately generated policy is. In other words, Procedure \ref{procedure_robust_policy_iteration} is not sensitive to small errors when the initial condition is in a neighbourhood of the optimal solution.

The next corollary removes the restrictive assumption that $\tilde P_0$ (resp. $\hat K_1$) needs to be in a neighbourhood of $P^*$  (resp. $K^*$). Its proof can be found in Appendix \ref{section_proof_corollary}.
\begin{corollary}\label{corollary_global}
For any $\epsilon>0$ and any given control gain $\hat{K}_1$ with $\mathcal{A}(\hat{K}_1)$ Hurwitz, there exist $\delta_4(\epsilon,\hat{K}_1)>0$, $\kappa^a_2(\delta_4)>0$, $\kappa^b_2(\delta_4)>0$, such that as long as
$\Vert \Delta G \Vert_\infty < \delta_4$, $[\hat{G}_{i}]_{uu}$ is invertible, $\mathcal{A}(\hat{K}_i)$ is Hurwitz, $\Vert \tilde{P}_i\Vert_F<\kappa^a_2$, $\Vert \hat{K}_i\Vert_F<\kappa^b_2$, $\forall i\in\mathbb{Z}_+$, $i>0$ and $\limsup_{i\rightarrow\infty} \Vert \tilde{P}_i-P^* \Vert_F<\epsilon$. If in addition $\lim_{i\rightarrow\infty} \Vert\Delta G_i\Vert_F = 0$, then $\lim_{i\rightarrow\infty} \Vert \tilde{P}_i-P^* \Vert_F=0$.
\end{corollary}
In Corollary \ref{corollary_global}, $\hat{K}_1$ can be any control gain whose $\mathcal{A}(\hat{K}_1)$ is Hurwitz, which is different from Theorem \ref{theorem_local_ISS}. When there is no error, i.e. $\Delta G_i \equiv 0$ for all $i=1,2,\cdots$, Corollary \ref{corollary_global} implies the convergence result of Procedure \ref{procedure_policy_itration} in Theorem \ref{theorem_standard_PI}.
\subsection{Convergence Analysis of OLSbPI}
Now we use Corollary \ref{corollary_global} to derive a convergent result of the OLSbPI algorithm. For given $\hat{K}_1$, let $\mathcal{K}$ denote the set of control gains (including $\hat{K}_1$) generated by Procedure \ref{procedure_robust_policy_iteration} with all possible $\{\Delta G_i\}_{i=1}^\infty$ satisfying $\Vert \Delta G\Vert_\infty<\delta_4$, where $\delta_4$ is the one in Corollary \ref{corollary_global}. The following result is firstly derived, whose proof can be found in Appendix \ref{section_proof_lemma_algorithm}.
\begin{lemma}\label{Lemma_policy_evaluation_error_bounded}
Under Assumptions \ref{assumption_moments} and \ref{PE_assumption}, there exist $s_0>0$ and $t_0>0$, such that for any $s_f\geq s_0$ and any $t_f\geq t_0$, $\hat{K}_i\in\mathcal{K}$ implies $\Vert\Delta G_i\Vert_F< \delta_4$, almost surely.
\end{lemma}
With Lemma \ref{Lemma_policy_evaluation_error_bounded} at hand, we are ready to prove the convergence of the OLSbPI.
\begin{theorem}\label{theorem_stochastic_algorithm_convergence}
In Algorithm \ref{algorithm_off_policy_stochastic}, under Assumptions \ref{assumption_moments} and \ref{PE_assumption}, for any initial control gain $\hat{K}_1$ with $\mathcal{A}(\hat{K}_1)$ Hurwitz and any $\epsilon>0$, there exist $s_0>0$ and $t_0>0$, such that for any $s_f\geq s_0$ and $t_f\geq t_0$, almost surely, $\limsup_{N\rightarrow\infty}\Vert\tilde{P}_N - P^*\Vert_F<\epsilon$
and $\mathcal{A}(\hat{K}_i)$ is Hurwitz for all $i=1,\cdots,N$, where $\tilde{P}_N$ is the unique solution of \eqref{algebraic_lyapunov_equation} for $\hat{K}_N$.
\end{theorem}
\begin{proof}
Since $\hat{K}_1\in\mathcal{K}$, Lemma \ref{Lemma_policy_evaluation_error_bounded} implies $\Vert\Delta G_1\Vert_F<\delta_4$ almost surely. By definition, $\hat{K}_2\in\mathcal{K}$. Thus $\Vert\Delta G_i\Vert_F<\delta_4, i=1,2,\cdots$ almost surely by mathematical induction. Then Corollary \ref{corollary_global} completes the proof.
\end{proof}
\begin{remark}
The optimal solution $P^*$ depends on the system matrices $\{D_j\}_{j=1}^{q_1}$, $\{F_k\}_{k=1}^{q_2}$ in the multiplicative noises. The change in $\{D_j\}_{j=1}^{q_1}$ or $\{F_k\}_{k=1}^{q_2}$ causes the change in $P^*$, then by Theorem \ref{theorem_stochastic_algorithm_convergence} the near-optimal solution found by OLSbPI algorithm changes accordingly. Thus the proposed OLSbPI algorithm is robust to the stochastic noises. This is an advantage not possessed by the off-policy RL algorithms in \cite{JIANG20122699,pang2020robust}.
\end{remark}
\section{Numerical Example}\label{section_simulation}
We consider a deterministic linear time-invariant system model of the triple inverted pendulum proposed in \cite{furut1984attitude}, and assume that it is perturbed by both multiplicative and additive noise\footnote{The code is available at \url{https://github.com/bo-pang/OLSbPI}}. Then the perturbed triple inverted pendulum can be described by system \eqref{LTI_sys_ito} with matrices $A\in\mathbb{R}^{6\times 6}$ and $B\in\mathbb{R}^{6\times 2}$ given in \cite[Section 3]{furut1984attitude}, $C=0.1I_6$ and
\begin{equation*}
\resizebox{\hsize}{!}{
    $(D_1)_{j,k} = \begin{cases}
               0.01, & \text{if }j=k=6,\\
               0, & \text{otherwise.}
            \end{cases},
    (F_1)_{j,k} = \begin{cases}
               0.01, & \text{if }j=4,k=1,\\
               0, & \text{otherwise.}
            \end{cases}.$
}
\end{equation*}
It is assumed that the parameters of all the system matrices are unknown, but an initial control gain
$$\hat{K}_1 = \left[\begin{array}{cccccc}
-9.44  & -3.11  & -1.2 &  -3.11 &  -1.31 &  -0.58\\
  -32.5 &  -11.51  &   -3.87  &  -10.72 &  -4.41 &  -2.01
\end{array}\right]$$
with $\mathcal{A}(\hat{K}_1)$ Hurwitz is available. Then the proposed OLSbPI algorithm is applied to find a near-optimal solution of the optimal stationary control problem with weighting matrices $Q=I_6$ and $R=I_2$. In the OLSbPI algorithm, we firstly apply the admissible control policy \eqref{input_collect_data_sys} with $\sigma_u = 100$ to the system to generate a trajectory of input/state data of length $t_f = 510$. Then data matrices $\psi_{t_f}$, $\zeta_{t_f}$ and $\xi_{t_f}$ are computed and policy evaluation step Line \ref{algorithm_policy_evaluation} in Algorithm \ref{algorithm_off_policy_stochastic} is implemented with $s_f=100$. The algorithm is terminated after $N=10$ iterations. The learning processes are illustrated in Fig. \ref{triple_inverted_pendulum}.
\begin{figure}[!htb]
    \centering
    \includegraphics[width=0.9\linewidth]{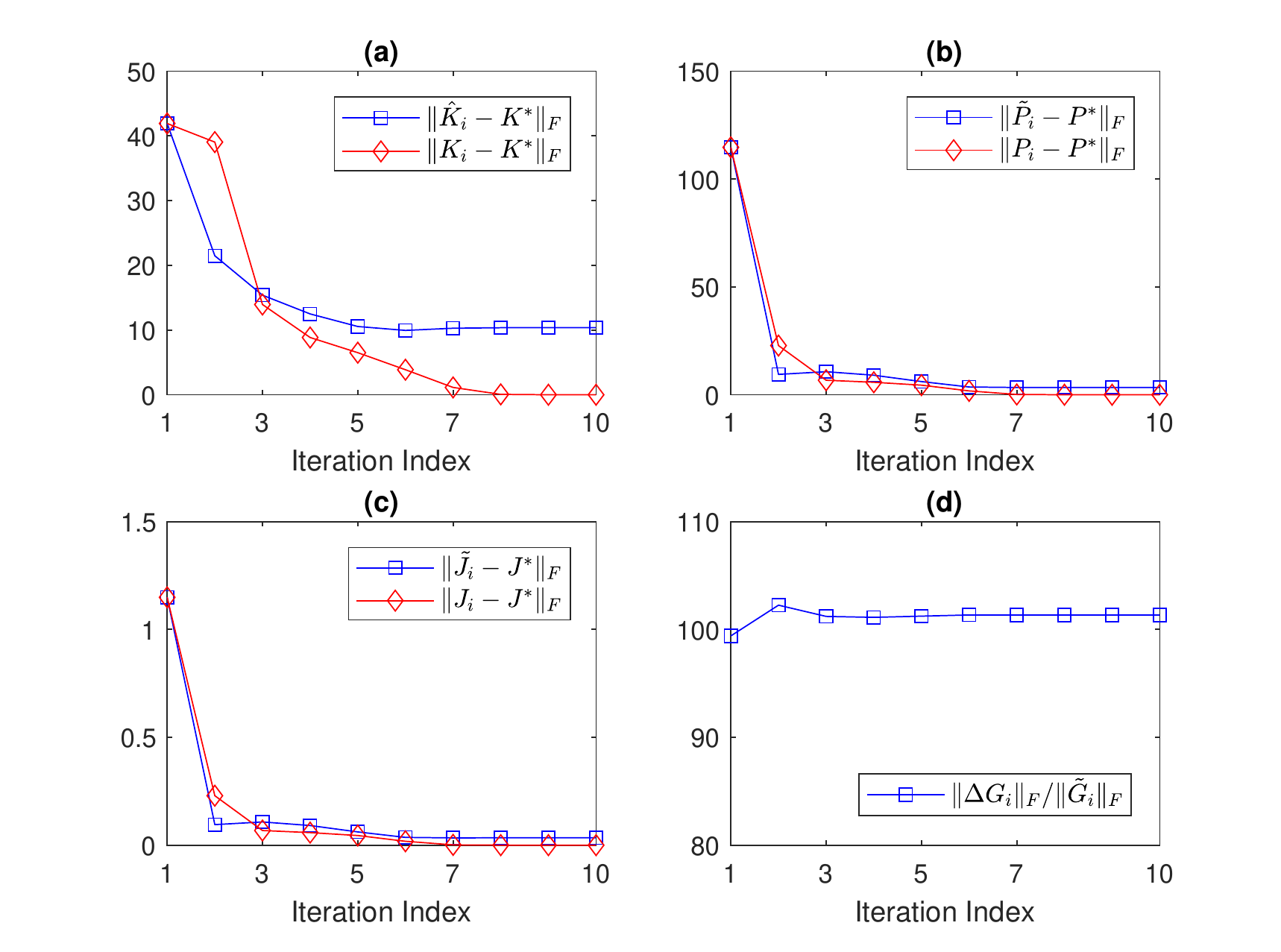}
    \caption{Learning processes of the OLSbPI algorithm on the triple inverted pendulum example.}\label{triple_inverted_pendulum}
\end{figure}
The trajectories of the norms of the differences between the control gain $\hat{K}_i$ generated by OLSbPI, its associated $\tilde{P}_i$, the true cost $\tilde{J}_i = \trace(C^T\tilde{P}_iC)$ it induced and the optimal values $K^*$, $P^*$ and $J^*$ are shown in Fig. \ref{triple_inverted_pendulum}-(a), Fig. \ref{triple_inverted_pendulum}-(b) and Fig. \ref{triple_inverted_pendulum}-(c), respectively. By Theorem \ref{theorem_standard_PI}, we run the model-based Procedure \ref{procedure_policy_itration} with $K_1=\hat{K}_1$ for a sufficiently large number of iterations, and use the results in the last iteration as the optimal values $K^*$, $P^*$ and $J^*$. The trajectories of the norms of the differences between the control gain $K_i$ generated by Procedure \ref{procedure_policy_itration}, its associated $P_i$, the true cost $J_i = \trace(C^TP_iC)$ it induced and the optimal values $K^*$, $P^*$ and $J^*$ are also drawn in Fig. \ref{triple_inverted_pendulum}. Fig. \ref{triple_inverted_pendulum}-(d) shows the trajectory of the estimation error $\Delta G_i$ relative to the true value $\tilde{G}_i$. One can see that although the estimation error $\tilde{G}_i$, caused by the unmeasurable stochastic noise in the system dynamics, distorts the trajectory generated by OLSbPI from the precise trajectory generated by the model-based Procedure \ref{procedure_policy_itration}, the OLSbPI still successfully finds a control policy with near-optimal cost. This is consistent with the convergence results we obtained in Corollary \ref{corollary_global} and Theorem \ref{theorem_stochastic_algorithm_convergence}.

The state trajectories of the closed-loop systems during the data collection and after the algorithm terminates has been shown in Fig. \ref{State_in_data_collection} and Fig. \ref{State_after_learning}, respectively. It can be observed in Fig. \ref{State_after_learning} that the state trajectories after the learning (the solid red lines) have smaller overshoot than those before the learning (the dashed blue lines).

\begin{figure}[!htb]
    \centering
    \includegraphics[width=\linewidth]{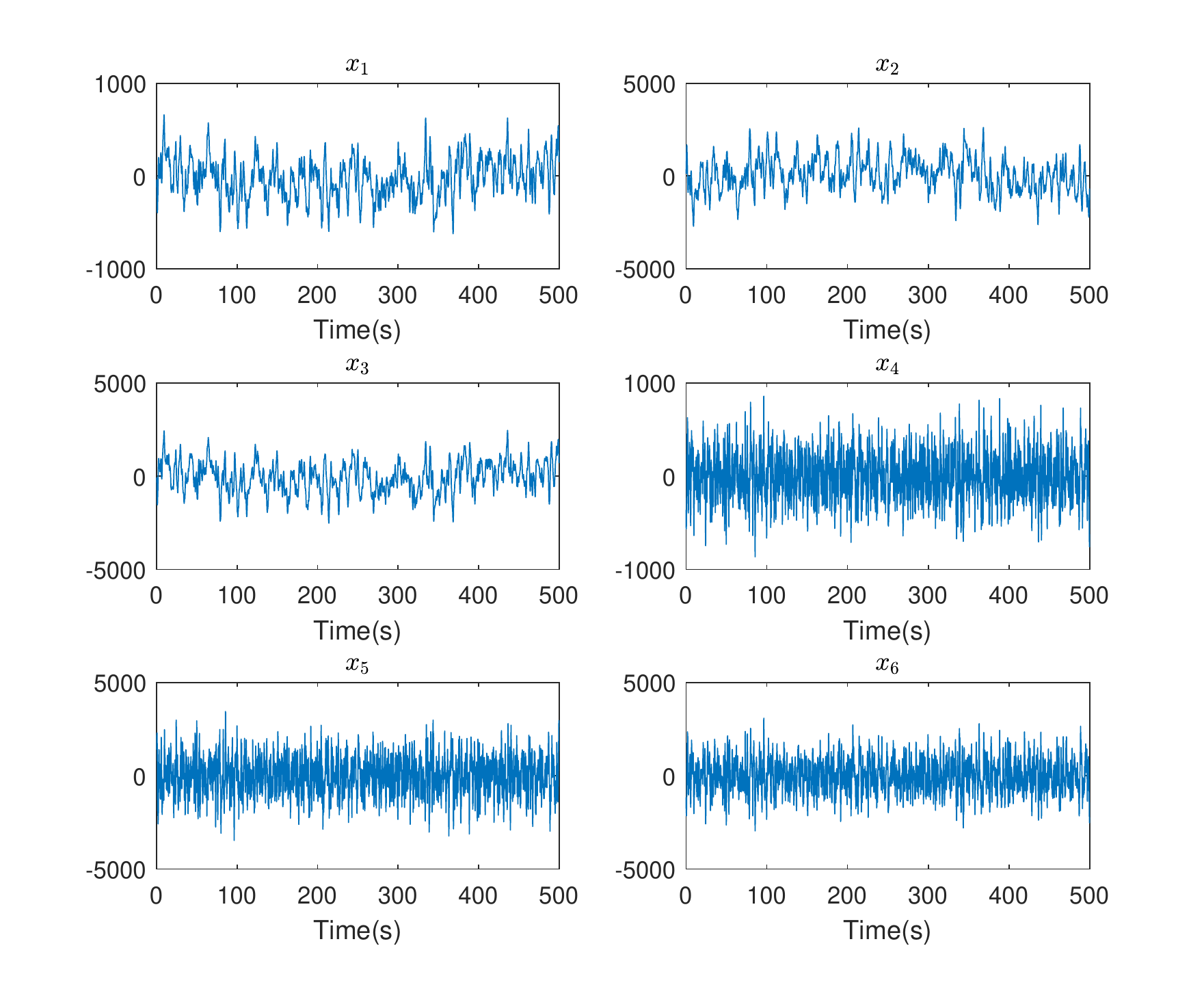}
    \caption{The state trajectories during the data collection phase.}\label{State_in_data_collection}
\end{figure}

\begin{figure}[!htb]
    \centering
    \includegraphics[width=\linewidth]{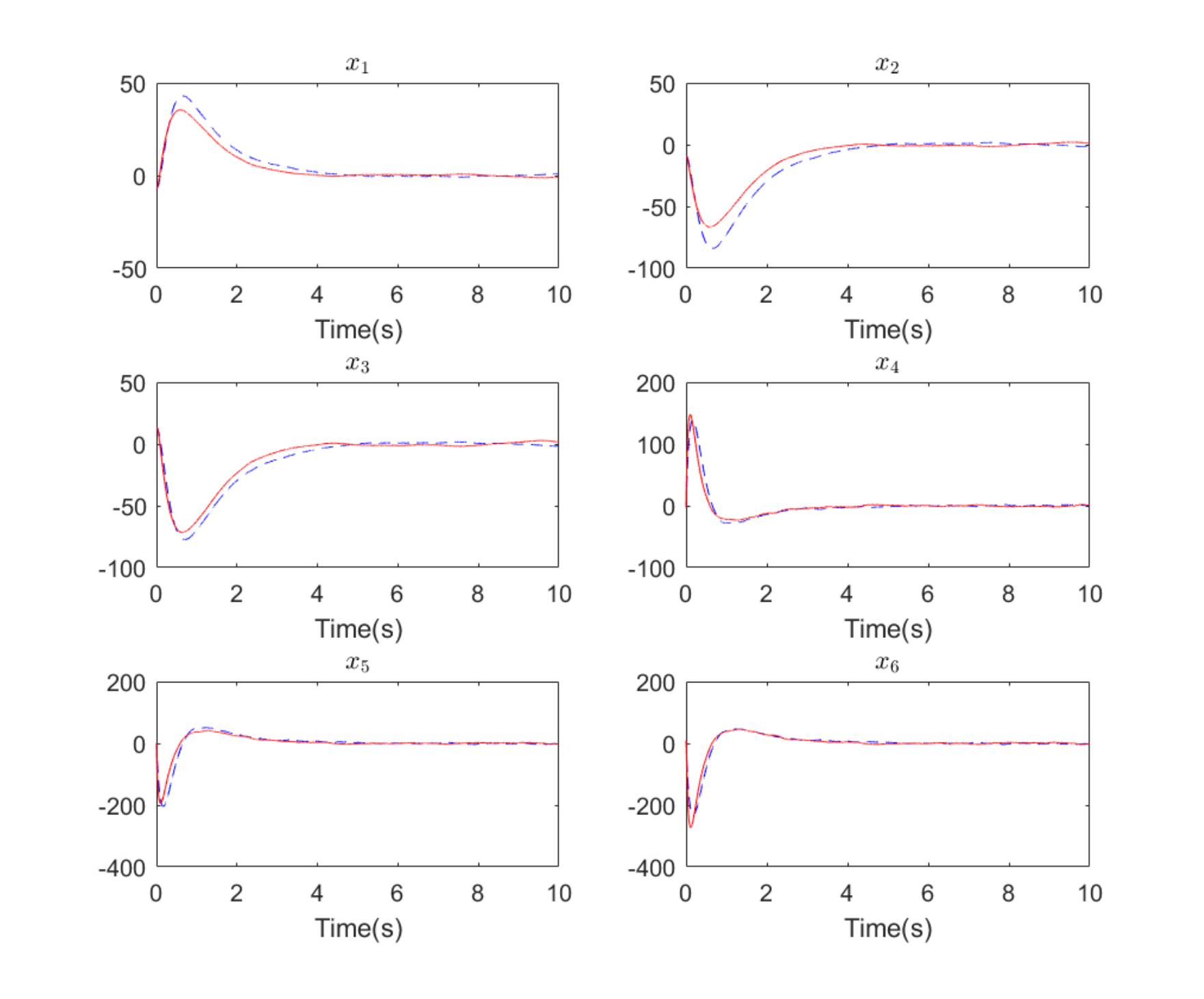}
    \caption{The state trajectories generated by the control policies before and after the learning. The solid red lines are state trajectories generated by control policy with $\hat{K}_{10}$, while the dashed blue lines are state trajectories generated by control policy with $\hat{K}_1$. }\label{State_after_learning}
\end{figure}

\section{Conclusion}\label{section_conclusion}
In this paper, we have proposed a novel data-driven optimistic least-squares-based policy iteration algorithm, to approximately solve the optimal stationary control problem for linear stochastic systems with additive and multiplicative noises, without requiring the precise knowledge of system matrices. Rigorous convergence analysis is given and shows that the proposed data-driven algorithm generates near-optimal solutions almost surely under suitable conditions, starting from an initial admissible control policy. The efficacy of the proposed optimal stationary control method has been validated by means of a triple inverted pendulum perturbed by stochastic noises.

\appendices
\section{Useful auxiliary results}\label{Appendix_auxiliary}
\begin{lemma}\label{uniform_exp_stab_over_compact_set}
Let $\mathcal{O}$ be a compact set of Huriwitz matrices, then there exist an $a_0>0$ and a $b_0>0$, such that $\Vert \exp(Ot)\Vert\leq a_0\exp(-b_0t)$ for any $O\in\mathcal{O}$.
\end{lemma}
\begin{proof}
For each $O\in\mathcal{O}$, by \cite[Theorem 3.20. and Corollary 3.6.]{ODE_dynamical_systems} there exist $r>0$, $a>0$ and $b>0$, such that $\Vert \exp(O't)\Vert\leq a\exp(-bt)$ for all $\Vert O'-O\Vert<r$. Then the compactness of $\mathcal{O}$ completes the proof.
\end{proof}
\begin{lemma}[{\cite[Page 57]{Magnus2007matrix}}]\label{lemma_relationship_vec_svec}
For $X\in\mathbb{S}^n$, there exists a unique matrix $D_n\in\mathbb{R}^{n^2\times \frac{1}{2}n(n+1)}$ with full column rank, such that
$\vect(X) = D_n\svec(X)$, $\svec(X) = D_n^\dagger \vect(X)$.
\end{lemma}
\begin{lemma}\label{math_fundamentals} The following results hold.
\begin{enumerate}
    \item\label{continuity_of_eigenvalues} \cite[Appendix D]{horn2012matrix} The eigenvalues of a square matrix depends continuously on the elements of the matrix.
    \item\label{continuity_of_inversion} \cite{Rakocevic1997} The matrix inverse of a square matrix is continuous function of the elements of the matrix.
    \item\label{extremum_value_theorem} \cite[Appendix E]{horn2012matrix} If $X$ is a compact set of a metric space and $f:X\rightarrow\mathbb{R}$ is a continuous function, then $f$ is bounded on $X$.
\end{enumerate}
\end{lemma}
For $X\in\mathbb{R}^{n\times n}$, $Y\in\mathbb{R}^{n\times m}$, $X+\Delta X\in\mathbb{R}^{n\times n}$, $Y + \Delta Y\in\mathbb{R}^{n\times m}$, supposing $X$ and $X + \Delta X$ are invertible, then
\begin{align}
    &\Vert X^{-1}Y - (X+\Delta X)^{-1}(Y+\Delta Y)\Vert_F\leq\left\Vert X^{-1}Y - \right.\nonumber\\
    &\left.X^{-1}(Y+\Delta Y)+ X^{-1}(Y+\Delta Y) - (X+\Delta X)^{-1}(Y+\Delta Y)\right\Vert_F\nonumber\\
    &\leq \Vert -X^{-1}\Delta Y + X^{-1}\Delta X(X + \Delta X)^{-1}(Y+\Delta Y)\Vert_F\nonumber\\
    &\leq \Vert X^{-1}\Vert_F\left(\Vert \Delta Y\Vert_F  \right.\nonumber\\
    &\left.+ \Vert(X + \Delta X)^{-1}\Vert_F\Vert(Y+\Delta Y)\Vert_F\Vert\Delta X\Vert_F\right).\label{linear_perturbation}
\end{align}
\section{Proof of Theorem \ref{theorem_local_ISS}}\label{section_proof_local_ISS}
Let $\delta_0$ be the one defined in Lemma \ref{lemma_invertible_hurwitz}. The following lemma can be proved by continuity in a similar way to that of \cite[Lemma 6]{pang2020robust_discrete}.
\begin{lemma}\label{lemma_stability_ISS}
For any $\tilde{P}_i\in\bar{\mathcal{B}}_{\delta_0}(P^*)$, there exists a $d(\delta_0)>0$, independent of $\tilde{P}_i$, such that $\mathcal{A}(\hat{K}_{i+1})$ is Hurwitz and $[\hat{G}_{i}]_{uu}$ is invertible, if $\Vert\Delta G_i\Vert_F\leq d$.
\end{lemma}
\begin{lemma}\label{lemma_total_error_bounded_ISS}
For any $\tilde{P}_i\in\bar{\mathcal{B}}_{\delta_0}(P^*)$ and any $\epsilon_3>0$, there exists a $0<\delta_2(\delta_0,\epsilon_3)\leq d$, independent of $\tilde{P}_i$, such that $\Vert \mathcal{E}(\tilde{G}_i,\Delta G_i)\Vert_2\leq C_2\Vert \Delta G_i\Vert_F<\epsilon_3$,
if $\Vert \Delta G_i\Vert_F<\delta_2$, where $C_2(\delta_2)>0$.
\end{lemma}
\begin{proof}
For any $\tilde{P}_i\in\bar{\mathcal{B}}_{\delta_0}(P^*)$, $\Vert \Delta G_i\Vert_F\leq d$, by \eqref{linear_perturbation} and Lemma \ref{math_fundamentals}
\begin{align}
    &\Vert \mathscr{K}(\tilde{P}_i) - \hat{K}_{i+1}\Vert_F
    \leq \Vert [\hat{G}_{i}]_{uu}^{-1}\Vert_F\left(\Vert\Delta G_i\Vert_F \nonumber \right.\\
    &\left.+ \Vert [\tilde{G}_{i}]_{uu}^{-1}\Vert_F \Vert[\tilde{G}_{i}]_{ux}\Vert_F\Vert \Delta G_i\Vert_F\right)  \leq C_3(\delta_0,d)\Vert \Delta G_i\Vert_F. \label{K_close}
\end{align}
Define $\Delta \mathcal{A}_i = \mathcal{A}(\mathscr{K}(\tilde{P}_i)) - \mathcal{A}(\hat{K}_{i+1})$ and
$$\Delta b_i = \vect\left(\mathscr{K}(\tilde{P}_i)^TR\mathscr{K}(\tilde{P}_i) - \hat{K}_{i+1}^TR\hat{K}_{i+1}\right).$$
Using \eqref{K_close}, it is easy to check that $\Vert \Delta \mathcal{A}_i\Vert_F\leq C_4\Vert\Delta G_i\Vert_F$, $\Vert \Delta b_i\Vert_2\leq C_5\Vert\Delta G_i\Vert_F$, for some $C_4(\delta_0,d)>0$, $C_5(\delta_0,d)>0$. Then by \eqref{linear_perturbation}, Lemma \ref{lemma_stability_ISS} and Lemma \ref{math_fundamentals}
\begin{equation*}
    \begin{split}
        &\Vert\mathcal{E}(\tilde{G}_i,\Delta G_i) \Vert_F \leq \left\Vert \mathcal{A}^{-1}\left(\hat{K}_{i+1}\right)\right\Vert_F \left(C_5 + C_4\left\Vert\mathcal{A}^{-1}(\mathscr{K}(\tilde{P}_i)\right\Vert_F\right.\\
        &\left.\times\left\Vert Q + \mathscr{K}(\tilde{P}_i)^TR\mathscr{K}(\tilde{P}_i)\right\Vert_F\right)\Vert \Delta G_i\Vert_F \leq C_2(\delta_0)\Vert \Delta G_i\Vert_F
    \end{split}
\end{equation*}
Choosing $0<\delta_2\leq d$ such that $C_2\delta_2<\epsilon_3$ completes the proof.
\end{proof}
\begin{proof}[Proof of Theorem \ref{theorem_local_ISS}]
Let $\epsilon_3=(1-\sigma)\delta_1$ in Lemma \ref{lemma_total_error_bounded_ISS}, and $\delta_3$ be equal to the $\delta_2$ associated with $\epsilon_3$. For any $i\in\mathbb{Z}_+$, if $\tilde{P}_i\in\mathcal{B}_{\delta_1}(P^*)$, then $[\hat{G}_i]_{uu}$ is invertible, $\mathcal{A}(\hat{K}_{i+1})$ is Hurwitz and
\begin{align}
    \Vert \tilde{P}_{i+1} - P^*\Vert_F&\leq \Vert \mathcal{E}(\tilde{G}_i,\Delta G_i)\Vert_2 +\left\Vert \mathcal{A}^{-1}(\mathscr{K}(\tilde{P}_i))\right.\nonumber\\
    &\left.\times\vect(-Q-\mathscr{K}(\tilde{P}_i)^TR\mathscr{K}(\tilde{P}_i)) - \vect(P^*) \right\Vert_2 \nonumber \\
    &\leq  \sigma\Vert \tilde{P}_i - P^*\Vert_F + C_2\Vert\Delta G\Vert_\infty \label{proof_theorem_local_ISS_inequality_2}\\
    &<\sigma\delta_1 + C_2\delta_3<\sigma\delta_1 + \epsilon_3 =\delta_1, \label{proof_theorem_local_ISS_inequality_3}
\end{align}
where the last two inequality are due to Lemmas \ref{lemma_local_exponentail_stable} and \ref{lemma_total_error_bounded_ISS}. By induction, \eqref{proof_theorem_local_ISS_inequality_2} and \eqref{proof_theorem_local_ISS_inequality_3} hold for all $i\in\mathbb{Z}_+$. Unrolling \eqref{proof_theorem_local_ISS_inequality_2} proves \ref{theorem_local_ISS_item_well_defined} and \ref{theorem_local_ISS_item_local_ISS} in Theorem \ref{theorem_local_ISS}. Then \eqref{K_close} implies \ref{theorem_local_ISS_item_K_bounded} in Theorem \ref{theorem_local_ISS}. \ref{theorem_local_ISS__item_converging_input_converging_state} can be proved in a similar way to that of (iii) in \cite[Theorem 2]{pang2020robust}.
\end{proof}

\section{Proof of Corollary \ref{corollary_global}}\label{section_proof_corollary}
Some auxiliary lemmas are presented before the proof of Corollary \ref{corollary_global}. For convenience, in the following we assume $Q>I_n$ and $R>I_m$. All the proofs still work for any $Q>0$ and $R>0$, by suitable rescaling.
\begin{lemma}\label{lemma_stability}
If $\mathcal{A}(\hat{K}_i)$ is Hurwitz, then $[\hat{G}_{i}]_{uu}$ is invertible and $\mathcal{A}(\hat{K}_{i+1})$ is Hurwitz, as long as $\Vert \Delta G_i\Vert_F < a_i$, where
\begin{equation*}
     a_i=\left(m(\sqrt{n} + \Vert \hat{K}_{i}\Vert_2)^2+m(\sqrt{n} + \Vert \hat{K}_{i+1}\Vert_2)^2\right)^{-1}.
\end{equation*}
Furthermore,
\begin{equation}\label{K_bound_by_P}
    \Vert\hat{K}_{i+1}\Vert_F \leq 2\Vert R^{-1}\Vert_F(1 + \Vert B^T\tilde{P}_i\Vert_F).
\end{equation}
\end{lemma}
\begin{proof}
By similar derivations to the proof of \cite[Lemma 6]{pang2020robust}, the invertibility of $[\hat{G}_i]_{uu}$ and inequality \eqref{K_bound_by_P} can be proved, and the following inequality is obtained
\begin{equation}\label{key_inequality_2}
    x^T\left(\mathcal{L}_{\hat{K}_{i+1}}(\tilde{P}_i) + Q+\hat{K}_{i+1}^TR\hat{K}_{i+1}\right)x - \epsilon_{2,i} \leq 0
\end{equation}
where for any $x\in\mathbb{R}^n$ on the unit ball
\begin{equation*}
\begin{split}
    \epsilon_{2,i} &= \Vert \Delta G_i\Vert_F \mathbf{1}^T(\vert \mathcal{X}_{\hat{K}_{i}}\vert_{abs}+\vert \mathcal{X}_{\hat{K}_{i+1}}\vert_{abs})\mathbf{1}<1,\\
    \mathcal{X}_{\hat{K}_i} &= \left[\begin{array}{c}
         I  \\
         -\hat{K}_i
    \end{array}\right]xx^T \left[\begin{array}{cc}
        I &  -\hat{K}_i^T
    \end{array}\right],\\
    \mathbf{1}^T\vert \mathcal{X}_{\hat{K}_i}\vert_{abs}\mathbf{1}&\leq m(\sqrt{n}+\Vert \hat{K}_i\Vert_2)^2,
\end{split}
\end{equation*}
and $\vert \mathcal{X}_{\hat{K}_{i}} \vert_{abs}$ denotes the matrix obtained from $\mathcal{X}_{\hat{K}_{i}}$ by taking the absolute value of each entry. Then $Q>I_n$ leads to $x^T\mathcal{L}_{\hat{K}_{i+1}}(\tilde{P}_i)x < 0$
for all $x$ on the unit ball. So $\mathcal{A}(\hat{K}_{i+1})$ is Hurwitz by Lemma \ref{lemma_lyapunov}.
\end{proof}
\begin{lemma}\label{lemma_boundedness}
For any $\bar{i}\in\mathbb{Z}_+$, $\bar{i}>0$,  if
\begin{equation}\label{error_condition_bounded}
  \Vert \Delta G_i\Vert_F< (1+i^2)^{-1}a_i,\quad i=1,\cdots, \bar{i},
\end{equation}
where $a_i$ is defined in Lemma \ref{lemma_stability}, then $\Vert\tilde{P}_i\Vert_F\leq 6\Vert \tilde{P}_1\Vert_F$, and $\quad \Vert\hat{K}_{i}\Vert_F\leq C_0$,
for $i=1,\cdots,\bar{i}$, where
$$ C_0 = \max\left\{ \Vert\hat{K}_1 \Vert_F, 2\Vert R^{-1}\Vert_F\left(1+6\Vert B^T\Vert_F\Vert \tilde{P}_1\Vert_F\right)\right\}.$$
\end{lemma}
\begin{proof}
Inequality (\ref{key_inequality_2}) yields
\begin{equation}\label{key_inequality_3}
    \mathcal{L}_{\hat{K}_{i+1}}(\tilde{P}_i) + (Q+\hat{K}_{i+1}^TR\hat{K}_{i+1}) - \epsilon_{2,i}I < 0.
\end{equation}
Inserting (\ref{eRPI_PE}) into above inequality, we have
\begin{equation*}
    \mathcal{L}_{\hat{K}_{i+1}}\left(\tilde{P}_{i} - \tilde{P}_{i+1}   - \mathcal{L}^{-1}_{\hat{K}_{i+1}}\left(\epsilon_{2,i}I\right)\right)<0.
\end{equation*}
By Lemma \ref{lemma_lyapunov},
\begin{equation}\label{key_inequality_4_1}
    \tilde{P}_{i+1} < \tilde{P}_{i} + \epsilon_{2,i}\mathcal{L}_{\hat{K}_{i+1}}^{-1}(-I).
\end{equation}
With $Q>I_n$, (\ref{key_inequality_3}) yields
\begin{equation*}
    \mathcal{L}_{\hat{K}_{i+1}}\left(\tilde{P}_i + (1 - \epsilon_{2,i})\mathcal{L}_{\hat{K}_{i+1}}^{-1}(I)\right) < 0.
\end{equation*}
By definition of $\epsilon_{2,i}$ and \eqref{error_condition_bounded}, $\epsilon_{2,i}<1$. Thus Lemma \ref{lemma_lyapunov} implies
\begin{equation}\label{key_inequality_4_2}
    \mathcal{L}_{\hat{K}_{i+1}}^{-1}(-I)<\frac{1}{1-\epsilon_{2,i}}\tilde{P}_i.
\end{equation}
From (\ref{key_inequality_4_1}) and (\ref{key_inequality_4_2}), we obtain
    $\tilde{P}_{i+1}<\left(1+\epsilon_{2,i}/(1-\epsilon_{2,i})\right)\tilde{P}_i$.
By definition of $\epsilon_{2,i}$ and condition \eqref{error_condition_bounded},
$\epsilon_{2,i}/(1-\epsilon_{2,i}) \leq 1/i^2$, $i=1,\cdots,\bar{i}$.
Then \cite[\S 28. Theorem 3]{InfinitSeries} yields
$\tilde{P}_i\leq 6\tilde{P}_1$, $i=1,\cdots,\bar{i}$.
An application of \eqref{K_bound_by_P} completes the proof.
\end{proof}
Note that all the conclusions of Corollary \ref{corollary_global} can be implied by Theorem \ref{theorem_local_ISS} if $\delta_4<\min(\gamma^{-1}(\epsilon),\delta_3)$, and $\tilde{P}_1\in\mathcal{B}_{\delta_1}(P^*)$ for Procedure \ref{procedure_robust_policy_iteration}. Thus the proof of Corollary \ref{corollary_global} reduces to the proof of the following lemma.
\begin{lemma}\label{lemma_converge_to_neighbourhood}
Given $\hat{K}_1$ with $\mathcal{A}(\hat{K}_1)$ Hurwitz, there exist $0<\delta_4<\min(\gamma^{-1}(\epsilon),\delta_3)$, $\bar{i}\in\mathbb{Z}_+$, $\kappa^a_3>0$, $\kappa^b_3>0$, such that $[\hat{G}_{i}]_{uu}$ is invertible, $\mathcal{A}(\hat{K}_i)$ is Hurwitz, $\Vert\tilde{P}_i\Vert_F< \kappa^a_3$, $\Vert \hat{K}_i\Vert_F< \kappa^b_3$, $i=1,\cdots,\bar{i}$, $\tilde{P}_{\bar{i}}\in\mathcal{B}_{\delta_1}(P^*)$, as long as $\Vert \Delta G\Vert_\infty<\delta_4$.
\end{lemma}
\begin{proof}
Consider Procedure \ref{procedure_robust_policy_iteration} confined to the first $\bar{i}$ iterations, where $\bar{i}$ is a sufficiently large integer to be determined later in this proof. Suppose
\begin{equation}\label{error_bound_2}
        \Vert\Delta G_i\Vert_F<b_{\bar{i}}\triangleq\frac{1}{2m(1+\bar{i}^2)}\left(\sqrt{n} + C_0\right)^{-2}.
\end{equation}
Condition (\ref{error_bound_2}) implies condition (\ref{error_condition_bounded}). Thus $\mathcal{A}(\hat{K}_i)$ is Hurwitz, $[\hat{G}_{i}]_{uu}^{-1}$ is invertible, $\Vert\tilde{P}_i\Vert_F$ and $\Vert\hat{K}_i\Vert_F$ are bounded. By (\ref{eRPI_PE}) we have
$\mathcal{L}_{\hat{K}_{i+1}}(\tilde{P}_{i+1}-\tilde{P}_{i}) = -Q-\hat{K}_{i+1}^TR\hat{K}_{i+1}-\mathcal{L}_{\hat{K}_{i+1}}(\tilde{P}_{i})$.
Letting $E_i = \hat{K}_{i+1} - \mathscr{K}(\tilde{P}_i)$, the above equation can be rewritten as
\begin{equation}\label{perturbed_newton}
    \tilde{P}_{i+1} = \tilde{P}_i - \mathcal{N}(\tilde{P_i}) + \mathcal{L}_{\mathscr{K}(\tilde{P_i})}^{-1}(\mathscr{E}_i),
\end{equation}
where $\mathcal{N}(\tilde{P_i}) = \mathcal{L}_{\mathscr{K}(\tilde{P_i})}^{-1}\circ\mathcal{R}(\tilde{P}_i),$ and 
\begin{equation*}
    \begin{split}
    \mathscr{E}_i &= - E^T_i\mathscr{R}(\tilde{P}_{i+1})E_i + E^T_i\mathscr{R}(\tilde{P}_{i+1})\left(\mathscr{K}(\tilde{P}_{i+1})-\mathscr{K}(\tilde{P}_i)\right) \\
    &+ \left(\mathscr{K}(\tilde{P}_{i+1})-\mathscr{K}(\tilde{P}_i)\right)^T\mathscr{R}(\tilde{P}_{i+1})E_i
    \end{split}
\end{equation*}
Given $\hat{K}_1$, let $\mathcal{M}_{\bar{i}}$ denote the set of all possible $\tilde{P}_i$, generated by (\ref{perturbed_newton}) under condition (\ref{error_bound_2}). By definition, $\{\mathcal{M}_j\}_{j=1}^\infty$ is a nondecreasing sequence of sets, i.e., $\mathcal{M}_1\subset \mathcal{M}_2 \subset \cdots$. Define $\mathcal{M} = \cup_{j=1}^\infty\mathcal{M}_j$, $\mathcal{D} = \{P\in \mathbb{S}^n\ \vert\ \Vert P\Vert_F\leq 6\Vert \tilde{P}_1\Vert_F\}$. Then by Lemma \ref{lemma_boundedness} and Theorem \ref{theorem_standard_PI}, $\mathcal{M}\subset\mathcal{D}$; $\mathcal{M}$ is compact; $\mathcal{A}(P)$ is Hurwitz for any $P\in\mathcal{M}$.\par
Now we prove that $\mathcal{N}(\cdot)$ is Lipschitz continuous on $\mathcal{M}$. By definition, \eqref{lyapunov_operator_equivalent}, \eqref{linear_perturbation} and Lemma \ref{math_fundamentals}
\begin{align}\label{Newton_Lipschitz_inequality}
        &\Vert\mathcal{N}(P^1)-\mathcal{N}(P^2)\Vert_F=\Vert\mathcal{A}^{-1}(\mathscr{K}(P^1))\vect(\mathcal{R}(P^1))\nonumber\\
        &- \mathcal{A}^{-1}(\mathscr{K}(P^2))\vect(\mathcal{R}(P^2))\Vert_2 \nonumber\\
        &\leq\Vert\mathcal{A}^{-1}(\mathscr{K}(P^1))\Vert_F\left(\Vert\mathcal{R}(P^1) - \mathcal{R}(P^2)\Vert_F\right. \nonumber\\
        &\left.+ \Vert\mathcal{A}^{-1}(\mathscr{K}(P^2))\Vert_F\Vert\mathcal{A}(\mathscr{K}(P^1))-\mathcal{A}(\mathscr{K}(P^2))\Vert_F\Vert\mathcal{R}(P^2)\Vert_F\right)  \nonumber\\
        &\leq \Vert\mathcal{A}^{-1}(\mathscr{K}(P^1))\Vert_F\left( C_6+ C_7 \Vert\mathcal{A}^{-1}(\mathscr{K}(P^2))\Vert_F\Vert\mathcal{R}(P^2)\Vert_F\right)\nonumber\\
        &\times\Vert P^1 - P^2\Vert_F \leq L_{\mathcal{N}} \Vert P^1 -P^2 \Vert_F,
\end{align}
where $C_6$ is the Lipschitz constant of $\mathcal{R}(\cdot)$ on compact set $\mathcal{D}$, $C_7$ is the Lipschitz constant of $\mathcal{A}(\mathscr{K}(\cdot))$ on compact set $\mathcal{M}$.

Define $\{P_{k\vert i}\}_{k=0}^{\infty}$ as the sequence generated by \eqref{Kleinman_nonlinear_system} with $P_{0\vert i}=\tilde{P}_i$.
Similar to \eqref{perturbed_newton}, we have
\begin{equation}\label{exact_newton}
    P_{k+1\vert i} = P_{k\vert i} - \mathcal{N}(P_{k\vert i}), \quad k\in\mathbb{Z}_+.
\end{equation}
By Theorem \ref{theorem_standard_PI} and the fact that $\mathcal{M}$ is compact, there exists $k_0\in\mathbb{Z}_+$, such that 
\begin{equation}\label{triangle_inequality_first_part}
\Vert P_{k_0\vert i}-P^*\Vert_F<\delta_1/2, \qquad \forall P_{0\vert i}\in\mathcal{M}.
\end{equation}
Suppose
\begin{equation}\label{perturb_small}
    \Vert\mathcal{L}_{\mathscr{K}(\tilde{P}_{i+j})}^{-1}(\mathscr{E}_{i+j})\Vert_F<\mu,\qquad j=0,\cdots, \bar{i} - i.
\end{equation}
Unrolling (\ref{perturbed_newton}) and (\ref{exact_newton}),
and using \eqref{Newton_Lipschitz_inequality} and (\ref{perturb_small}), we have 
\begin{equation*}
    \begin{split}
        \Vert P_{k\vert i} - \tilde{P}_{i+k}\Vert_F \leq k\mu + \sum_{j=0}^{k-1} L_{\mathcal{N}} \Vert P_{j\vert i} - \tilde{P}_{i+j}\Vert_F.
    \end{split}
\end{equation*}
An application of the Gronwall inequality \cite[Theorem 4.1.1.]{agarwal2000difference} to the above inequality implies
\begin{equation}\label{continuous_dependence}
    \Vert P_{k\vert i} - \tilde{P}_{i+k}\Vert_F \leq k\mu + L_{\mathcal{N}}\mu \sum_{j=0}^{k-1}j (1+L_{\mathcal{N}})^{k-j-1}.
\end{equation}
The error term in (\ref{perturbed_newton}) satisfies
\begin{align}
    \left\Vert \mathcal{L}_{\mathscr{K}(\tilde{P_i})}^{-1}(\mathscr{E}_i)\right\Vert_F 
    &= \left\Vert \mathcal{A}^{-1}(\tilde{P}_i) \vect\left( \mathscr{E}_i\right)\right\Vert_2 \nonumber \leq C_8\Vert \mathscr{E}_i\Vert_F,\label{perturbed_inverse_lypunov}
\end{align}
where $C_8>0$ is obtained by Lemma \ref{math_fundamentals}. Let $\bar{i}>k_0$, and $k=k_0$, $i = \bar{i}-k_0$ in \eqref{continuous_dependence}. Then by condition \eqref{error_bound_2}, Lemma \ref{lemma_boundedness}, (\ref{perturb_small}), and (\ref{continuous_dependence}), there exists $i_0\in\mathbb{Z}_+$, $i_0>k_0$, such that $\Vert P_{k_0\vert \bar{i}-k_0} -\tilde{P}_{\bar{i}}\Vert_F<\delta_1/2$, for all $\bar{i}\geq i_0$. Setting $i = \bar{i}-k_0$ in \eqref{triangle_inequality_first_part}, the triangle inequality yields $\tilde{P}_{\bar{i}}\in\mathcal{B}_{\delta_1}(P^*)$, for $\bar{i}\geq i_0$. Then in \eqref{error_bound_2}, choosing $\bar{i}\geq i_0$ such that $\delta_4 = b_{\bar{i}}<\min(\gamma^{-1}(\epsilon),\delta_3)$ completes the proof.
\end{proof}

\section{Proof of Lemma \ref{Lemma_policy_evaluation_error_bounded}}\label{section_proof_lemma_algorithm}
\begin{proof}
By definition,
\begin{equation*}
    \begin{split}
        \Vert\Delta G_i\Vert_F \leq \Vert \hat{G}_i - G(\hat{P}_i(s_f))\Vert_F + \Vert G(\hat{P}_i(s_f)) - G(\tilde{P}_i)\Vert_F.
    \end{split}
\end{equation*}
Thus the task is to prove that each term in the right-hand side of the above inequality is less than $\delta_4/2$. To this end, we firstly study $\Vert \hat{P}_i(s_f)-\tilde{P}_i\Vert_F$. Letting $\hat{p}_i=\vect(\hat{P}_i)$, by Lemma \ref{lemma_relationship_vec_svec}, Line \ref{algorithm_policy_evaluation} in Algorithm \ref{algorithm_off_policy_stochastic} can be rewritten as
\begin{equation}\label{ode_policy_evaluation_data_estimate_2}
\begin{split}
    \dot{\hat{p}}_{i} = \mathcal{T}^1(\psi_{t_f},\zeta_{t_f},\hat{K}_i)\hat{p}_{i} + \mathcal{T}^2(\psi_{t_f},\xi_{t_f},\hat{K}_i),
\end{split}
\end{equation}
where $\hat{p}_i(0)\in\mathbb{R}^{n^2}$ and
\begin{equation*}
\begin{split}
&\mathcal{T}^1(\psi_{t_f},\zeta_{t_f},\hat{K}_i) = \Gamma_i D_{m+n+1}\psi_{t_f}^\dagger\zeta_{t_f}D_n^\dagger,    \\
&\mathcal{T}^2(\psi_{t_f},\xi_{t_f},\hat{K}_i) = \Gamma_i D_{m+n+1}\psi_{t_f}^\dagger\xi_{t_f}, \\
&\Gamma_i = \left(\left[I_n,-\hat{K}^T_i\right]\otimes\left[I_n,-\hat{K}^T_i\right]\right)\left(\left[I_{m+n},0\right]\otimes\left[I_{m+n},0\right]\right).
\end{split}
\end{equation*}
For convenience, in the sequel we define
\begin{equation*}
    \begin{split}
        \mathcal{T}^1_{t_f,i} &= \mathcal{T}^1(\psi_{t_f},\zeta_{t_f},\hat{K}_i),\quad \mathcal{T}^2_{t_f,i} = \mathcal{T}^2(\psi_{t_f},\xi_{t_f},\hat{K}_i),  \\
        \mathcal{T}^1_i &= \mathcal{T}^1(\Psi,\mathcal{Z},\hat{K}_i),\quad \mathcal{T}^2_i = \mathcal{T}^2(\Psi,\Xi,\hat{K}_i).
    \end{split}
\end{equation*}
Similar arguments applied to \eqref{ode_policy_evaluation_data_precise} with $K=\hat{K}_i$ yield
\begin{equation}\label{ode_policy_evaluation_exact_estimation}
    \dot{\bar{p}}_{i} =
    \mathcal{T}^1_i\bar{p}_{i} + \mathcal{T}^2_i,\quad \bar{p}_{i}(0)\in\mathbb{R}^{n^2}.
\end{equation}
Since \eqref{ode_policy_evaluation_data_precise} is identical to \eqref{ode_policy_evaluation_model}, \eqref{ode_policy_evaluation_exact_estimation} is identical to \eqref{ode_policy_evaluation_model_vector} with $K$ and $\vect(P)$ replaced by $\hat{K}_i$ and $\bar{p}_i$ respectively, and
\begin{equation}\label{data_model_identical}
    \mathcal{T}^1_i=\mathcal{A}(\hat{K}_i),\quad\mathcal{T}^2_i = \vect(Q+\hat{K}_i^TR\hat{K}_i).
\end{equation}
Let $s$ denote the time variable in \eqref{ode_policy_evaluation_data_estimate_2} and \eqref{ode_policy_evaluation_exact_estimation}, since $\mathcal{A}(\hat{K}_i)$, $\hat{K}_i\in\mathcal{K}$ is Hurwitz, by Lemma \ref{Lemma_optimistic_PI}
\begin{equation}\label{policy_evaluation_converge}
    \lim_{s\rightarrow\infty} \bar{P}_{i}(s) = \tilde{P}_i
\end{equation} 
where $\bar{P}_i = \vect^{-1}(\bar{p}_i)$ and $\tilde{P}_i$ is the unique solution of \eqref{algebraic_lyapunov_equation} with $K=\hat{K}_i$. Let $\mathcal{V}$ be the set of the unique solutions of \eqref{algebraic_lyapunov_equation} with $K\in\mathcal{K}$. Then by Corollary \ref{corollary_global} $\mathcal{V}$ and $\mathcal{K}$ are bounded, and $\mathcal{A}(K)$ is Hurwitz, for $\forall K\in\mathcal{K}$. So $\mathcal{A}(K)$ is Hurwitz, for $\forall K\in\bar{\mathcal{K}}$. Otherwise by definition and continuity of the determinant of a matrix, there must be a limit point of set $\mathcal{K}$, denoted as $K_{\text{lim}}$, such that $\mathcal{A}(K_{\text{lim}})$ is singular. Then let $K^{(k)}$ be a sequence of control gains in $\mathcal{K}$ converging to $K_{\text{lim}}$. Again by continuity of determinant of a matrix, we have $\lim_{K\rightarrow\infty}\det(K^{(k)}) = \det(K_{\text{lim}}) = 0$. Let $P^{(k)}$ be the solution of \eqref{algebraic_lyapunov_equation} with $K=K^{(k)}$. Then by \eqref{lyapunov_operator_equivalent} and \eqref{algebraic_lyapunov_equation} we have
$\lim_{k\rightarrow\infty}\Vert P^{(k)}\Vert_F = \infty$
which contradicts the boundedness of $\mathcal{V}$. Define $\mathcal{K}_1 = \{\mathcal{A}(K)\vert K\in\bar{\mathcal{K}}\}$. By continuity, $\mathcal{K}_1$ is a compact set of Hurwitz matrices, and there exists a $\delta_5>0$, such that any $X\in\bar{\mathcal{K}}_2$ is Hurwitz, where $\mathcal{K}_2 = \{X\vert X\in\mathcal{B}_{\delta_5}(Y),Y\in\mathcal{K}_1\}$.
Define
\begin{equation*}
\begin{split}
    \Delta\mathcal{T}^1_{t_f,i} = \mathcal{T}^1_i - \mathcal{T}^1_{t_f,i},\quad \Delta\mathcal{T}^2_{t_f,i} =  \mathcal{T}^2_i - \mathcal{T}^2_{t_f,i}.
\end{split}    
\end{equation*}
The boundedness of $\mathcal{K}$, \eqref{ergodic_converge_1}, \eqref{ergodic_converge_2} and \eqref{data_model_identical} imply the existence of $t_1>0$, such that for any $t_f\geq t_1$, any $\hat{K}_i\in\mathcal{K}$, almost surely
\begin{equation}\label{belong_to_compact_stable_set}
    \mathcal{T}^1_{t_f,i}\in \bar{\mathcal{K}}_2, \quad \mathcal{T}^2_{t_f,i}< C_9,
\end{equation}
where $C_9>0$ is a constant. Then $ \mathcal{T}^1_{t_f,i}$ is Hurwitz and \eqref{ode_policy_evaluation_data_estimate_2} admits a unique stable equilibrium, that is,
\begin{equation}\label{data_policy_evaluation_converge}
    \lim_{s\rightarrow\infty}\hat{P}_{i}(s) = \mathring{P}_i
\end{equation}
for some $\mathring{P}_i\in\mathbb{S}^n$. From \eqref{ode_policy_evaluation_data_estimate_2}, \eqref{ode_policy_evaluation_exact_estimation}, \eqref{policy_evaluation_converge} and \eqref{data_policy_evaluation_converge}, we have
\begin{equation*}
    \begin{split}
        \vect(\tilde{P}_i) = -\left( \mathcal{T}^1_i \right)^{-1}\mathcal{T}^2_i, \quad
        \vect(\mathring{P}_i) = -\left(\mathcal{T}^1_{t_f,i} \right)^{-1}\mathcal{T}^2_{t_f,i}.
    \end{split}
\end{equation*}
Thus by \eqref{linear_perturbation}, for any $t_f\geq t_1$ and any $\hat{K}_i\in\mathcal{K}$, almost surely
\begin{align*}
&\Vert \mathring{P}_i - \tilde{P}_i\Vert_F\leq \left\Vert\left(\mathcal{T}^1_i \right)^{-1}\right\Vert_F\left(\Vert \Delta\mathcal{T}^2_{t_f,i} \Vert_F +\right.\\
&\left.\left\Vert\left(\mathcal{T}^1_{t_f,i} \right)^{-1}\right\Vert_F\left\Vert\mathcal{T}^2_{t_f,i}\right\Vert_2\Vert\Delta\mathcal{T}^1_{t_f,i} \Vert_F\right) \\
&\leq C_{10} \Vert \Delta\mathcal{T}^2_{t_f,i} \Vert_F + C_{11}\Vert\Delta\mathcal{T}^1_{t_f,i} \Vert_F
\end{align*}
where $C_{10}$ and $C_{11}$ are some positive constants, and the last inequality is due to \eqref{data_model_identical}, \eqref{belong_to_compact_stable_set} and the fact that $\mathcal{K}_1$ and $\bar{\mathcal{K}}_2$ are compact sets of Hurwitz matrices. Then for any $\epsilon_1>0$, the boundedness of $\mathcal{K}$, \eqref{ergodic_converge_1} and \eqref{ergodic_converge_2} imply the existence of $t_2\geq t_1$, such that for any $t_f\geq t_2$, almost surely
\begin{equation}\label{PE_error_first_half}
    \Vert \mathring{P}_i - \tilde{P}_i\Vert_F<\epsilon_1/2,
\end{equation}
as long as $\hat{K}_i\in\mathcal{K}$. By Lemma \ref{uniform_exp_stab_over_compact_set} and \eqref{PE_error_first_half}, for any $t_f\geq t_2$ and any $\hat{K}_i\in\mathcal{K}$,
$\Vert \mathring{P}_{i}-\hat{P}_i(s)\Vert_F\leq a_0\exp(-b_0s)\Vert \mathring{P}_i\Vert_F\leq a_1\exp(-b_0s)$,
for some $a_0>0$, $b_0>0$ and $a_1>0$. Therefore there exists a $s_1>0$, such that for any $s_f\geq s_1$ and any $t_f\geq t_2$, almost surely
\begin{equation}\label{PE_error_second_half}
    \Vert \hat{P}_{i}(s_f)-\mathring{P}_i\Vert_F<\epsilon_1/2,
\end{equation}
as long as $\hat{K}_i\in\mathcal{K}$. With \eqref{PE_error_first_half} and \eqref{PE_error_second_half}, we obtain
\begin{equation}\label{triangle_inequality_3}
    \Vert\hat{P}_{i}(s_f)-\tilde{P}_i\Vert_F<\epsilon_1,
\end{equation}
almost surely for any $s_f\geq s_1$, any $t_f\geq t_2$, as long as $\hat{K}_i\in\mathcal{K}$. Since $\epsilon_1$ is arbitrary, we can choose $\epsilon_1$ such that almost surely
$\Vert G(\hat{P}_{i}(s_f)) - G(\tilde{P}_i)\Vert_F<\delta_4/2$
for any $s_f\geq s_1$ and any $t_f\geq t_2$, as long as $\hat{K}_i\in\mathcal{K}$.

Secondly, since $\mathcal{V}$ is bounded, by \eqref{triangle_inequality_3} $\hat{P}_{i}(s_f)$ is also almost surely bounded. Thus from Lines \ref{algorithm_estimate_final_Q_1} to \ref{algorithm_estimate_final_Q_2} in Algorithm \ref{algorithm_off_policy_stochastic}, \eqref{ergodic_converge_1} and \eqref{ergodic_converge_2}, there exists $t_3\geq t_2$, such that
$\Vert \hat{G}_{i} - G(\hat{P}_{i}(s_f))\Vert_F<\delta_4/2$
for any $t_f\geq t_3$ and any $s_f\geq s_1$, as long as $\hat{K}_i\in\mathcal{K}$.

Setting $s_0 = s_1$ and $t_0 = t_3$ yields $\Vert\Delta G_i\Vert<\delta_4$.
\end{proof}

\bibliography{RPIbib} 
\bibliographystyle{ieeetr}

\end{document}